\RequirePackage[l2tabu, orthodox]{nag}
\RequirePackage{etoolbox}

\newtoggle{showoverflow}
\newtoggle{allowtodo}
\newtoggle{showtodo}
\newtoggle{anonymous}
\newtoggle{submission}
\newtoggle{fullpage}
\newtoggle{toctitlepage}
\newtoggle{warningtodo}

\toggletrue{allowtodo}
\toggletrue{showtodo}

\toggletrue{fullpage}
\toggletrue{toctitlepage}

\newcommand{\remove}[1]{}

\documentclass[envcountsame,envcountsect,runningheads]{llncs}

\iftoggle{showoverflow}{
  \overfullrule=10pt
}{}

\iftoggle{fullpage}{
\usepackage{fullpage}
}{}

\makeatletter
\newcommand*{\toccontents}{\@starttoc{toc}}
\makeatother

\usepackage[utf8]{inputenc}
\usepackage[T1]{fontenc}
\usepackage{lmodern}
\usepackage{microtype}

\usepackage{amsmath}
\usepackage{amssymb}
\usepackage{mathtools}
\usepackage{mathrsfs}

\usepackage{booktabs}
\usepackage{array}
\newcolumntype{L}[1]{>{\raggedright\arraybackslash}m{#1}}
\newcolumntype{C}[1]{>{\centering\let\newline\\\arraybackslash\hspace{0pt}}m{#1}}
\newcolumntype{R}[1]{>{\raggedleft\arraybackslash}m{#1}}
\usepackage{threeparttable}
\usepackage{blkarray}
\usepackage{multirow}

\usepackage[noadjust]{cite}
\usepackage{url}
\usepackage{xcolor}
\usepackage{xspace}
\usepackage{needspace}
\usepackage[inline]{enumitem}

\usepackage{pifont}

\usepackage{tikz}
\usetikzlibrary{arrows.meta,calc,positioning,fit,backgrounds,shapes.geometric,decorations.pathreplacing}

\usepackage{float}
\usepackage{graphicx}
\usepackage{caption,subcaption}
\usepackage{framed}
\usepackage{mdframed}
\mdfsetup{skipabove=0pt,skipbelow=0pt}

\usepackage{braket}

\usepackage[pdfpagelabels=true,pagebackref]{hyperref}

\hypersetup{
  linktoc = page,
  pdfpagemode = UseNone,
  colorlinks,
  linkcolor={red!50!black},
  citecolor={blue!50!black},
  urlcolor={blue!80!black}
}

\renewcommand*{\backref}[1]{}
\renewcommand*{\backrefalt}[4]{
  \ifcase #1
  (Not cited.)
  \or
  (Page~#2.)
  \else
  (Pages~#2.)
  \fi
}

\makeatletter
\AtBeginDocument{
  \hypersetup{
    pdftitle = {\@title},
  }
}
\makeatother

\usepackage[capitalize]{cleveref}

\crefname{conjecture}{Conjecture}{Conjectures}
\Crefname{conjecture}{Conjecture}{Conjectures}

\usepackage{algorithm}
\usepackage[noend]{algpseudocode}

\expandafter\patchcmd\csname\string\algorithmic\endcsname
  {\labelwidth 0.5em}{\labelwidth0pt\labelsep0pt}{}{}

\makeatletter
\newcounter{algorithmicH}
\let\oldalgorithmic\algorithmic
\renewcommand{\algorithmic}{
  \stepcounter{algorithmicH}
  \oldalgorithmic}
\renewcommand{\theHALG@line}{ALG@line.\thealgorithmicH.\arabic{ALG@line}}
\makeatother

\renewcommand{\eqref}[1]{\mbox{Equation~(\ref{#1})}}

\AtEndEnvironment
{proof}
{\qed}

\newsavebox{\fboxenvbox}

\iftoggle{showoverflow}{
  \overfullrule=10pt
}{}

\newcommand{\etal}{et al.\xspace}

\newcommand{\C}{\mathbb{C}}
\newcommand{\F}{\mathbb{F}}

\newcommand{\N}{\mathbb{N}}

\newcommand*{\prob}[1]{{\Pr}\left[\,{#1}\,\right]}

\newcommand*{\Tr}[1]{{\mathsf{Tr}}\left[\,{#1}\,\right]}
\newcommand*{\Ptr}[2]{{\mathsf{Tr}_{#1}}\left[\,{#2}\,\right]}

\newcommand*{\Se}[1]{{\mathsf{S}}\left[\,{#1}\,\right]}

\newcommand{\calA}{\mathcal{A}}
\newcommand{\calC}{\mathcal{C}}

\newcommand{\bfm}[1]{{\normalfont\bfseries #1}}

\makeatletter
\newcommand\suchthat{
 \@ifstar
  {\mathrel{}\middle|\mathrel{}}
  {\mid}
}
\makeatother

\newcommand{\defeq}{\coloneqq}

\algdef{SnE}[Oracle]{Oracle}{EndOracle}
  [1]{\underline{#1}}
  {}

\algdef{SnE}[AlgoExp]{AlgoExp}{EndAlgoExp}
  [1]{#1}
  {}

\newcommand{\ve}{\varepsilon}
\newcommand{\eps}{\ve}

\DeclareMathOperator{\negl}{negl}

\newcommand{\QNCz}{\ensuremath{\mathsf{QNC^0}}\xspace}

\newcommand{\QACzf}{\ensuremath{\mathsf{QAC^0_f}}\xspace}

\newcommand{\bits}{{\{0,1\}}}

\DeclareMathOperator{\poly}{\mathsf{poly}}
\DeclareMathOperator{\polylog}{\mathsf{polylog}}
\DeclareMathOperator{\Bern}{\mathsf{Bern}}
\DeclareMathOperator{\Unif}{\mathsf{Unif}}

\newcommand{\bA}{\mathbf{A}}
\newcommand{\bC}{\mathbf{C}}
\newcommand{\bS}{\mathbf{S}}
\newcommand{\bE}{\mathbf{E}}
\newcommand{\bo}{\mathbf{o}}
\newcommand{\bp}{\mathbf{p}}
\newcommand{\br}{\mathbf{r}}
\newcommand{\by}{\mathbf{y}}
\newcommand{\bx}{\mathbf{x}}
\newcommand{\bz}{\mathbf{z}}
\newcommand{\bw}{\mathbf{w}}
\newcommand{\be}{\mathbf{e}}
\newcommand{\bs}{\mathbf{s}}
\newcommand{\bu}{\mathbf{u}}
\newcommand{\bv}{\mathbf{v}}

\makeatletter
\newcommand\mynobreakpar{\par\nobreak\@afterheading}
\makeatother

\providecommand{\orcidID}{}
\renewcommand{\orcidID}[1]{
	\!
	\raisebox{-1pt}{
		\href{https://orcid.org/#1}{\protect\includegraphics[width=9pt]{orcid.pdf}}
	}
}

\usepackage{nicefrac}
\usepackage{comment}
\usepackage[many]{tcolorbox}
\usepackage{xparse}
\usepackage{xcolor}
\definecolor{caribbeangreen}{rgb}{0.0, 0.8, 0.6}
\definecolor{calpolypomonagreen}{rgb}{0.12, 0.3, 0.17}
\definecolor{darkcyan}{rgb}{0.0, 0.55, 0.55}
\definecolor{cyan(process)}{rgb}{0.0, 0.72, 0.92}
\definecolor{shamrockgreen}{rgb}{0.0, 0.62, 0.38}
\iftoggle{allowtodo}{
  \iftoggle{showtodo}{
    \newcommand{\todo}[1]{\textcolor{red}{TODO: #1}\iftoggle{warningtodo}{\PackageWarning{TODO:}{#1!}}{}}
    \newcommand{\andru}[1]{\textcolor{orange}{AG: #1}\iftoggle{warningtodo}{\PackageWarning{AG:}{#1!}}{}}

  }{
    \newcommand{\todo}[1]{}
    \newcommand{\andru}[1]{}

  }
}

\title{Pseudoentanglement in constant depth \texorpdfstring{\\}{: } {\large How trivial states can have non-trivial entanglement structure}}
\titlerunning{Pseudoentanglement in constant depth}

\date{}

\iftoggle{anonymous}{
  \author{}
  \institute{}
}{
	\author{
		Alexandru Gheorghiu
	}
	\hypersetup{
		pdfauthor = {Alexandru Gheorghiu}
	}

	\authorrunning{A.\ Gheorghiu}

	\institute{
		IBM Research \\
		\href{mailto:agheorghiu@ibm.com}{agheorghiu@ibm.com}
	}
}

\usepackage[strings]{underscore}

\newcommand*{\appref}[1]{Appendix~\ref{#1}}

\begin{document}

{\def\addcontentsline#1#2#3{}
\renewcommand{\rm}{\textrm}
\maketitle}

\nottoggle{submission}{
  \iftoggle{toctitlepage}{
    \vspace{-0.25cm}
  }{}
}{}

\sloppy

\begin{abstract}
  We construct a family of 2D-local constant-depth quantum circuits that output states whose entanglement entropy across a specified cut cannot be estimated in quantum polynomial time.
  As constant-depth quantum circuits can be learned from polynomially many quantum samples, our resulting \emph{pseudoentangled} states are implicitly \emph{public-key} and not \emph{pseudorandom}. This separates pseudoentanglement from pseudorandomness in the shallow-circuit regime: the former is possible, while the latter is not.
  The construction is based on the quantum intractability of the \emph{Dense-Sparse Learning Parity with Noise} problem introduced in~\cite{dao2024lossy} and uses a bounded-fan-in, bounded-fan-out classical \emph{randomized encoding} for linear maps $\bx\mapsto \mathbf{N}\bx,$ which could be of independent interest.
  As applications, we obtain quantum hardness for the problem of learning the entanglement structure (across a fixed cut) of the ground-state of 1D and 2D local Hamiltonians. The 1D Hamiltonian has an inverse polynomial gap, whereas the 2D one has a constant gap. This complements the result of~\cite{groundstatehardness} that showed only factoring-based hardness for the 1D case, though achieving a volume versus area entanglement difference.
\end{abstract}

\section{Introduction}
\label{sec:intro}
The study of entanglement is central to quantum information and computation, where entanglement is treated as a resource. Estimating how much entanglement is produced by a given process is therefore a natural and important problem. Given a description of a quantum circuit, a quantum computer can efficiently determine whether the output pure state is close to a product state across a given bipartition (or cut)~\cite{v011a003}. On the other hand, estimating the entanglement entropy itself appears to be hard, even for quantum computers~\cite{v006a003}. This motivates the notion of \textit{pseudoentanglement}: states whose entanglement structure is computationally hard to determine~\cite{ITCS:ABFGVZ24,gheorghiu}.

Formally, a family of $n$-qubit pure states $\{\ket{\psi_n}\}_{n>0}$ is pseudoentangled if there does not exist a quantum polynomial-time (QPT) algorithm that, given $\poly(n)$ copies of $\ket{\psi_n}$, can decide whether $\ket{\psi_n}$ has entanglement entropy, across a specified cut, below $a$ or above $b$, where $b-a\ge 1/\poly(n)$, with probability non-negligibly greater than $1/2$.\footnote{One can also consider stronger notions of pseudoentanglement in which the gap between high and low entanglement is $\Omega(n)$ versus $O(\log n)$ and is required across every cut, rather than across a single fixed cut~\cite{ITCS:ABFGVZ24}.}

Pseudoentangled states can be viewed as a weaker version of \emph{quantum computational pseudorandomness}: the only property that is hidden is the amount of entanglement entropy. The more standard notions of quantum computational pseudorandomness are \emph{pseudorandom states} (PRSs)~\cite{JiLiuSong18}, which are efficiently generated keyed families of pure states that are computationally indistinguishable from Haar-random states given polynomially many copies, and \emph{pseudorandom unitaries} (PRUs)~\cite{MetgerPSY24,MaHuang24}, which are efficiently generated keyed families of unitaries that are computationally indistinguishable from Haar-random unitaries to polynomial-time distinguishers with oracle access to those unitaries. Both PRSs and PRUs have found several cryptographic applications, including encryption-, commitment- and signature-type primitives~\cite{C:AnaQiaYue22,C:MorYam22}. They have also appeared in recent discussions of quantum gravity: Harlow and Hayden linked computational difficulty to the black-hole information paradox, and subsequent works used pseudorandomness to argue that the so-called AdS/CFT dictionary, a mapping between states and observables in a quantum field theory and those in a quantum gravity theory, is exponentially complex~\cite{HarlowHayden,AaronsonBarbados16,ITCS:BouFefVaz20}. A recurring concern, however, is that the use of pseudorandom states or unitaries in this way makes it so that this complexity is effectively ``hidden behind a black-hole event horizon''~\cite{SusskindHCT20}. Pseudoentanglement was, in part, motivated by trying to address this concern~\cite{gheorghiu,ITCS:ABFGVZ24,AkersBCKMV24}. In addition to these potential applications in quantum gravity research, pseudoentanglement has also been used to show hardness of entropy estimation and general entanglement-learning tasks~\cite{gheorghiu,ITCS:ABFGVZ24,publickey,groundstatehardness,ChengFI25}.

As alluded to above, one interesting, and perhaps surprising, feature of pseudoentangled states is that \emph{they need not be pseudorandom}. Any pseudorandom state construction requires that the circuit preparing the states is hidden from the observer. In~\cite{gheorghiu,publickey}, it was shown that one can instead construct pseudoentangled families whose preparation circuits are public and of polynomial size. Such states are distinguishable from Haar-random states, which necessarily have exponential circuit complexity. These are referred to as \emph{public-key} pseudoentangled states. They allow one to prove hardness results in a stronger model than is usually available for either pseudorandom states or unitaries. For pseudorandom states and unitaries, hardness is usually expressed in a \emph{sample complexity} or \emph{black-box} model where one is either given copies of some unknown state or is allowed queries to an oracle and the goal is to show that there are no algorithms that can distinguish the pseudorandom objects from Haar random ones with polynomially-many samples, or queries, respectively. By contrast, public-key pseudoentanglement allows one to work in the \emph{computational complexity} model, where the input is a classical description of a quantum state, such as a circuit or a Hamiltonian, and hardness is shown based on standard computational assumptions, such as the hardness of factoring or Learning with Errors (\textsf{LWE}).

Given this distinction between pseudorandmness and pseudoentanglement, some natural questions that arise are: how far can we push the separation between the two notions? Could there be a large discrepancy in the complexity required to prepare pseudorandom states, compared to pseudoentangled states? To that end, we consider the \emph{circuit depth} required for producing these states. In terms of upper bounds, it is known that circuits of logarithmic depth, assuming all-to-all connectivity, suffice for both pseudorandom and pseudoentangled states. More recently, it has been shown that one can construct PRUs and PRSs in depth $\polylog(n)$ in one dimension and depth $\poly(\log\log n)$ with all-to-all connectivity~\cite{SchusterHH24}, and even in constant-time models such as \QACzf---informally, constant-depth quantum circuits augmented with stronger nonlocal operations such as multi-control Toffoli gates and multi-target CNOT gates~\cite{FoxmanPVY26}.
On the other hand, the result of~\cite{shallowlearning} shows that constant-depth quantum circuits (with bounded fan-in gates) can always be learned efficiently given polynomially many copies of their output states when acting on $\ket{0^n}$. This rules out constructions of both pseudorandom states and pseudorandom unitaries in the standard bounded fan-in circuit class \QNCz---that is, polynomial-size constant-depth circuits built from one- and two-qubit gates.
Arguably, this is not surprising, as constant-depth circuits are very limited in the states they can prepare. It is known, for instance, that they cannot produce highly nonlocal states such as GHZ or W states~\cite{watts}, which are constructible in \QACzf. In condensed matter physics, states that can be mapped to a product state by a constant-depth local circuit are referred to as \emph{trivial states}, and they all belong to the same quantum phase of matter. However, the fact that pseudoentangled states need not have their preparation circuits kept secret begs the question

\begin{quote}
    \centering \emph{Can pseudoentangled states be prepared by constant depth geometrically local quantum circuits?}
\end{quote}

We show that the answer is yes, by giving a construction of 2D geometrically local constant-depth circuits that produce pseudoentangled states, based on plausible post-quantum cryptographic assumptions.\footnote{We note that version 1 of~\cite{gheorghiu} claimed that constant-depth circuits can produce states with hard-to-estimate entanglement entropy. While this is technically correct, that construction implicitly assumed the availability of unbounded fan-out gates, a point clarified in version 2. Our construction does \emph{not} use unbounded fan-out; it requires only standard one- and two-qubit gates. In terms of complexity classes, the distinction is that the earlier result used $\mathsf{QNC^0_f}$ circuits, whereas the one here uses \QNCz (in fact, a geometrically local version of \QNCz). See \Cref{def:faninout} for the locality terminology used throughout.} Specifically, the construction is based on Dense-Sparse Learning Parities with Noise, a code-based post-quantum assumption introduced by Dao and Jain~\cite{dao2024lossy}. The main technical ingredient, which we isolate as a separate theorem of independent interest, is a \emph{perfect randomized encoding} of a dense linear map with bounded fan-in and bounded fan-out. The encoding preserves the entropy of that mapping, up to a fixed additive correction. By then choosing appropriate linear mappings whose output entropy under a random input is either large or small, and applying the randomized encodings of those mappings to a collection of product states, we obtain the desired pseudoentangled states.
The locality of the randomized encoding makes it so that the circuits preparing these states are in \QNCz.

We then leverage this result to prove hardness consequences in the local-Hamiltonian input model. First, by conjugating one-local projectors with our \QNCz preparation circuit, we obtain a simple family of 2D frustration-free constant-gap local Hamiltonians whose unique ground states inherit the entanglement gap. Second, using a 1D Feynman--Kitaev history-state construction, we obtain 1D local Hamiltonians with inverse-polynomial gap whose unique ground states retain the same fixed-cut entanglement gap up to a small additive loss. This gives a 1D Hamiltonian hardness result from a post-quantum code-based assumption. We compare this to the result of Bouland, Zhang, and Zhou~\cite{groundstatehardness}, that also examined the hardness of estimating entanglement entropy of ground states of local Hamiltonians. The differences between our result and theirs are that we only achieve a small entanglement difference for a fixed cut rather than their stronger near-area-law versus near-volume-law difference for geometric cuts. However, our 1D hardness is based on the post-quantum Dense-Sparse \textsf{LPN} assumption, whereas theirs is based on factoring, which is quantumly tractable.

\subsection{Main results}

Our main result is the construction of publicly samplable families of 2D-local constant-depth circuits on $n$ qubits whose output states have different entanglement entropy across a single fixed cut $W\subseteq[n]$, while the public circuit descriptions remain computationally indistinguishable. More formally we have:

\begin{theorem}[Constant-depth pseudoentanglement]
    \label{thm:mainintro}
    Assuming the intractability of Dense-Sparse \textsf{LPN} for quantum polynomial-time (QPT) adversaries, for all sufficiently large $n>0$ there exists an efficiently samplable distribution over tuples
    \[
        (\calC_n^{low},\calC_n^{high},W_n),
    \]
    where each $\calC_n^b$ is a 2D-local constant-depth circuit on $n$ qubits and $W_n\subseteq[n]$, such that:
    \begin{enumerate}
        \item \bfm{Additive entanglement entropy gap.}
        There exists a function $\Delta:\N\to\mathbb{R}_{\ge 1}$ such that, with overwhelming probability over the sampled tuple, letting
        \[
            \ket{\psi_n^{low}}=\calC_n^{low}\ket{0^n},
            \qquad
            \ket{\psi_n^{high}}=\calC_n^{high}\ket{0^n},
        \]
        and
        \[
            \rho_n^b=\Ptr{W_n}{\ket{\psi_n^b}\bra{\psi_n^b}},
            \qquad b\in\{low,high\},
        \]
        one has
        \[
            \Se{\rho_n^{high}}
            \ge
            \Se{\rho_n^{low}}+\Delta(n).
	        \]
		    Moreover, the gap can be chosen so that $\Delta(n)=\omega(\log n)$.

        \vspace{0.5em}

	        \item \bfm{Indistinguishability.}
        For every QPT algorithm $\calA$,
        \[
            \left|
                \prob{1\leftarrow \calA[\calC_n^{high},W_n]}
                -
                \prob{1\leftarrow \calA[\calC_n^{low},W_n]}
            \right|
            =
            \negl(n),
        \]
        where the probability is over the sampled tuple and the internal randomness of $\calA$.
    \end{enumerate}
\end{theorem}

The subset $W_n$ specifies the cut across which the entanglement is measured and need not be a connected region of the lattice.
The theorem above rests on a separate \emph{randomized encoding} statement that may be of independent interest.

A randomized encoding of a function $f$ is a different function, $\hat{f},$ such that each output $f(x)$ of the original function is associated to $\hat{f}(x, r),$ where $r$ is a random string~\cite{FOCS:IshKus00}. An efficient decoder can recover $f(x)$ given $\hat{f}(x, r),$ but no other information about $x$ is revealed. Essentially, the encoding allows one to represent a computation using randomness, so that the encoded version preserves the output of the original computation while hiding everything else about the input. Crucially, randomized encodings can be much easier to compute than the original computation. It is this fact which is key to our result.

\begin{theorem}[Bounded-fan-in, bounded-fan-out randomized encoding of dense linear maps]
    \label{thm:reintro}
    For every matrix $\mathbf{N}\in\F_2^{q\times m}$, there is an explicit perfect randomized encoding
    \[
        \mathsf{RE}_{\mathbf{N}}(\bx;\br,\bs)=(\bw,\bz)
    \]
    of the map $\bx\mapsto \mathbf{N}\bx$ with the following properties:
    \begin{enumerate}
        \item \bfm{Perfect decodability and privacy.}
        There is a deterministic decoder that recovers $\mathbf{N}\bx$ from $\mathsf{RE}_{\mathbf{N}}(\bx;\br,\bs)$, and there is an efficient simulator whose output distribution on input $\mathbf{N}\bx$ is exactly the distribution of $\mathsf{RE}_{\mathbf{N}}(\bx;R,S)$ for uniform masks $R$ and $S$.
        \item \bfm{Bounded fan-in and bounded fan-out.}
        Every output bit depends on at most three source bits, and every source bit influences at most three output bits.
        \item \bfm{Exact entropy formula.}
        For every random variable $X$ over $\bits^m$, if $R\leftarrow\bits^{mq}$ and $S\leftarrow\bits^{q(m-1)}$ are uniform and independent of $X$, then
        \[
            H(\mathsf{RE}_{\mathbf{N}}(X;R,S))
            =
            mq+q(m-1)+H(\mathbf{N}X).
        \]
    \end{enumerate}
\end{theorem}

The last item is what we leverage for pseudoentanglement: the randomized encoding adds only $\mathbf{N}$-independent mask entropy, so the low-versus-high entropy gap of the underlying linear map is preserved exactly.
As a corollary, the randomized encoding also yields a bounded-fan-in, bounded-fan-out collision-resistant hash from a random-code bounded syndrome-decoding/bounded shortest-vector assumption; see \appref{app:crh}.

We can use the construction from \Cref{thm:mainintro} to obtain hardness-of-learning statements for the entanglement structure of geometrically local Hamiltonians. Starting with the 2D case, given a circuit $\calC$, define
\[
    H_{\calC}=\sum_{i=1}^n \calC\ket{1}\!\bra{1}_i\calC^\dagger,
\]
where $\ket{1}\!\bra{1}_i$ denotes a projector onto the $\ket{1}$ state for qubit $i$, while acting as identity on all other qubits.
This Hamiltonian is unitarily equivalent to $\sum_i\ket{1}\!\bra{1}_i$, so its unique ground state is $\calC\ket{0^n}$ and its spectral gap is $1$. If $\calC$ is 2D-local and has constant depth, every term in $H_{\calC}$ has constant support on the lattice. As this is exactly the type of circuit we have from \Cref{thm:mainintro}, we get

\begin{theorem}[2D Hamiltonian entanglement learning]
    \label{thm:hamintro}
    Assuming the intractability of Dense-Sparse \textsf{LPN} for QPT adversaries, for all sufficiently large $n>0$ there exists an efficiently samplable distribution over tuples
    \[
        (H_n^{low},H_n^{high},W_n),
    \]
    where each $H_n^b$ is a 2D $k$-local Hamiltonian on $n$ qubits for some constant $k>0$ and $W_n\subseteq[n]$, such that:
    \begin{enumerate}
        \item \bfm{Ground-state structure.} With overwhelming probability over the sampled tuple, each Hamiltonian $H_n^b$ is frustration-free, has a unique ground state $\ket{\phi_n^b}$, and has spectral gap $1$.
        \item \bfm{Entanglement entropy difference.} There exists a function $\Delta:\N\to\mathbb{R}_{\ge 1}$ with $\Delta(n)=\omega(\log n)$ such that, with overwhelming probability over the sampled tuple,
        \[
            \Se{\Ptr{W_n}{\ket{\phi_n^{high}}\bra{\phi_n^{high}}}}
            \ge
            \Se{\Ptr{W_n}{\ket{\phi_n^{low}}\bra{\phi_n^{low}}}}
            +
            \Delta(n).
        \]
        \item \bfm{Indistinguishability.} For every QPT algorithm $\calA$,
        \[
            \left|
                \prob{1\leftarrow \calA[H_n^{high},W_n]}
                -
                \prob{1\leftarrow \calA[H_n^{low},W_n]}
            \right|
            =
            \negl(n).
        \]
    \end{enumerate}
\end{theorem}

We can also obtain a similar result for 1D local Hamiltonians. Since the circuits constructing the pseudoentangled states are not 1D local, we instead use a standard 1D history-state construction.

\begin{theorem}[1D Hamiltonian entanglement learning]
    \label{thm:ham1dintro}
    Assuming the intractability of Dense-Sparse \textsf{LPN} for QPT adversaries, for all sufficiently large $n>0$ there exists an efficiently samplable distribution over tuples
    \[
        (H_n^{low},H_n^{high},W_n),
    \]
    where each $H_n^b$ is a 1D $k$-local Hamiltonian on $n$ qubits for some constant $k>0$ and $W_n\subseteq[n]$, such that:
    \begin{enumerate}
        \item \bfm{Ground-state structure.} With overwhelming probability over the sampled tuple, each Hamiltonian $H_n^b$ is frustration-free, has a unique ground state $\ket{\phi_n^b}$, and has spectral gap at least $1/\poly(n)$.
        \item \bfm{Entanglement entropy difference.} There exists a function $\Delta:\N\to\mathbb{R}_{> 0}$ with $\Delta(n)\ge 1/2$ such that, with overwhelming probability over the sampled tuple,
        \[
            \Se{\Ptr{W_n}{\ket{\phi_n^{high}}\bra{\phi_n^{high}}}}
            \ge
            \Se{\Ptr{W_n}{\ket{\phi_n^{low}}\bra{\phi_n^{low}}}}
            +
            \Delta(n).
        \]
        \item \bfm{Indistinguishability.} For every QPT algorithm $\calA$,
        \[
            \left|
                \prob{1\leftarrow \calA[H_n^{high},W_n]}
                -
                \prob{1\leftarrow \calA[H_n^{low},W_n]}
            \right|
            =
            \negl(n).
        \]
    \end{enumerate}
\end{theorem}

\subsection{Overview of techniques}
\label{subsec:overview}

The starting point is the Dense-Sparse \textsf{LPN} lossy-function
pair of Dao and Jain~\cite{dao2024lossy}. Consider the matrix
\[
    \bA=\mathbf{T}\mathbf{M}\in\F_2^{d/2\times m},
\]
where $\mathbf{M}\in\F_2^{d\times m}$ has constant-weight sparse columns and
$\mathbf{T}\in\F_2^{d/2\times d}$ is uniform.  Let
$\bS\in\F_2^{r\times d/2}$, let
$\bE\leftarrow\Bern(\varepsilon)^{r\times m}$, and let
$\bC\in\F_2^{r\times m}$ be a robust sparse-recovery matrix. On sparse
inputs $\bx\in\bits^m$ define
\[
    F^{\mathrm{low}}(\bx)=(\bA\bx,(\bS\bA\oplus\bE)\bx),
    \qquad
    F^{\mathrm{high}}(\bx)
      =(\bA\bx,(\bS\bA\oplus\bE\oplus\bC)\bx).
\]
The claim is that the two functions are computationally indistinguishable from each other, with $F^{\mathrm{low}}$ being many-to-one (or \emph{lossy}) and $F^{\mathrm{high}}$ being injective.\footnote{We use boldface notation for vectors and
matrices, and nonbold uppercase letters for random variables.}

Both maps are only ever evaluated on inputs of low Hamming weight.  The input
distribution is $X\sim\Bern(t/m)^{\otimes m}$, so $\mathbb{E}[|X|]=t$ and a
Chernoff bound places $X$ in the Hamming ball
\[
    T=\{\bx\in\bits^m:|\bx|\le 2t\},
    \qquad
    \ell=\log_2|T|,
\]
except with negligible probability.  It is therefore enough to understand the
two maps on $T$, and lossiness of $F^{\mathrm{low}}$ means that the image
$F^{\mathrm{low}}(T)$ is much smaller than $T$ itself.

The reason $F^{\mathrm{low}}$ collapses $T$ is that its first output block cannot see
any more of $\bx$ than the sparse product $\mathbf{M}\bx$ does.  Let $k$ denote
the constant column weight of $\mathbf{M}$.  Each nonzero coordinate of $\bx$
switches on a single column of $\mathbf{M}$, and that column carries only $k$
nonzero entries, so
\[
    |\bx|\le 2t
    \qquad\Longrightarrow\qquad
    |\mathbf{M}\bx|\le 2kt;
\]
the sparse product never leaves a Hamming ball of radius $2kt$ in $\F_2^{d}$.
Since $\bA\bx=\mathbf{T}(\mathbf{M}\bx)$, the first output block factors through that
ball, and hence takes at most
\[
    \sum_{j\le 2kt}\binom dj
\]
distinct values on $T$, no matter how many inputs are fed to it.  This single
counting bound is the entire source of the compression, and it is exactly where
sparsity is needed: for a dense $\mathbf{M}$ the product $\mathbf{M}\bx$ would be
unconstrained and no such bound would hold.

Whether this is a genuine compression is a question about parameters.  Writing
$m=d^{a}$ and $t=\Theta(d^{\delta})$ as in \Cref{conj:densesparselpn}, the count
above is $2^{\Theta(kt\log(d/t))}=2^{\Theta(k(1-\delta)\,t\log d)}$, whereas
$\ell=\Theta(t\log(m/t))=\Theta((a-\delta)\,t\log d)$.  The two exponents are in
ratio $\Theta\!\left(k(1-\delta)/(a-\delta)\right)$, so pushing $\delta$ towards
$1$ makes the image of the first output block an arbitrarily small power of $|T|$.
Rearranging, the inequality $k(1-\delta)\le(a-\delta)/D$ is equivalent to
$\delta\ge\frac{Dk-a}{Dk-1}$, which is exactly the admissible window for
$\delta$ in \Cref{conj:densesparselpn}; choosing $D$ large enough as a function
of $\Gamma$ therefore bounds the image of the first output block by $|T|^{1/D}$.

The second output block adds little on top of this.  Once $\bA\bx$ is fixed, so is
$\bS\bA\bx$, and the only freedom left in $(\bS\bA\oplus\bE)\bx$ is the noise
term $\bE\bx$.  A Chernoff bound together with a union bound over $T$ shows
that, for a constant $\gamma>0$, the bound $|\bE\bx|\le\gamma r$ holds
simultaneously for every $\bx\in T$ except with negligible probability.  Given
the first output block, the second output block therefore contributes at most
$\sum_{j\le\gamma r}\binom rj$ further values.  Here $r=\Theta(\ell)$, and
$\gamma$ can be made an arbitrarily small constant by shrinking the constant in
$\varepsilon=\Theta(1/t)$, so this is again only a small power of $|T|$.
Multiplying the two counts gives
\[
    |F^{\mathrm{low}}(T)|\le|T|^{1/\Gamma}
\]
for any prescribed constant $\Gamma>1$.  Since $F^{\mathrm{low}}$ maps all $|T|$
inputs into this image, its average preimage has size at least
$|T|^{1-1/\Gamma}$, which is the sense in which the low branch is many-to-one.

Adding $\bC$ destroys the argument above, and that is the point of the high
branch: $\bC$ is built independently of $\mathbf{M}$, so the second output block no
longer factors through the small set of sparse products.  Indeed, the trapdoor
$\bS$ reveals $\bC\bx\oplus\bE\bx$, and on the same noise event as above a
syndrome decoder followed by an outer error-correcting decoder recovers every
input in $T$.  Hence the high output retains essentially all the entropy $H(X).$ Conversely, the lossiness of the low branch has support smaller by an arbitrarily prescribed
constant exponent and so its output will have significantly less entropy than the input.

Let $\mathbf{N}^b$ be the row-stacked matrix defining $F^b$, i.e.
\[
    \mathbf{N}^{\mathrm{low}}
    =
    \begin{bmatrix}
        \bA\\
        \bS \bA \oplus \bE
    \end{bmatrix},
    \qquad
    \mathbf{N}^{\mathrm{high}}
    =
    \begin{bmatrix}
        \bA\\
        \bS \bA \oplus \bE \oplus \bC
    \end{bmatrix}
\]
Consider the state
\[
    \sum_{\bx\in\bits^m}\sqrt{\prob{X=\bx}}\ket{\bx}\ket{\mathbf{N}^b\bx}.
\]
It has entanglement entropy exactly $H(\mathbf{N}^bX)$ across the input-output cut. The arguments above, proved in \Cref{lem:densesparsepair}, imply that
\[
    H(\mathbf{N}^{high}X)-H(\mathbf{N}^{low}X)=\Omega(\ell),
\]
where $\ell=\log_2|T|$ is as above.\footnote{The lossiness and robust
inversion estimates are uniform over the entire ball $T$; see
\Cref{sec:qubitproof}.}

The problem is that the above state is not one that can be prepared in constant depth. This is because each output bit of $\mathbf{N}^b\bx$ is a parity of many input bits, and each input bit participates in many such parities. A naive coherent implementation would therefore require logarithmic depth or unbounded fan-out. However, the main observation is that \emph{we do not need to compute $\mathbf{N}^b\bx$ explicitly}. It is enough to sample a randomized encoding of it, and to do so in a way that preserves entropy exactly.

Let us now describe a randomized encoding for $\bx \mapsto \mathbf{N}^b \bx$ that can be computed in constant depth with bounded fan-in and bounded fan-out.
For each input bit $\bx_j$, we introduce uniform mask bits $\br_{j,1},\ldots,\br_{j,q}$ and output the bits
\[
    \bw_{j,0}=\bx_j\oplus \br_{j,1},
    \qquad
    \bw_{j,i}=\br_{j,i}\oplus \br_{j,i+1},
\]
with $i \in [q-1].$ Also, for each row $i$, we introduce the uniform masks $\bs_{i,1},\ldots,\bs_{i,m-1}$ with boundary convention $\bs_{i,0}=\bs_{i,m}=0$, and output
\[
    \bz_{i,j}=\bs_{i,j-1}\oplus \bs_{i,j}\oplus \mathbf{N}_{i,j}^b\br_{j,i}.
\]
Notice that every output bit depends on at most three source bits, and every source bit feeds only constantly many outputs. The global decoder recovers $\mathbf{N}^b\bx$ by taking row parities of the $\bz$-strings and correcting them using the $\bw$-strings. More importantly for us, if $R$ and $S$ are uniform masks, then
\[
    H\bigl(\mathsf{RE}_{\mathbf{N}^b}(X;R,S)\bigr)
    =
    mq+q(m-1)+H(\mathbf{N}^bX).
\]
The first two terms are just the mask entropy and are identical in the low and high branches. Hence the entropy difference is exactly $H(\mathbf{N}^{high}X)-H(\mathbf{N}^{low}X)=\Omega(\ell).$

We can now get a quantum circuit for preparing pseudoentangled states, by simply applying the randomized encoding circuit coherently on a collection of single-qubit states. Specifically, each $X_j$ becomes
\[
    \sqrt{1-t/m}\ket{0}+\sqrt{t/m}\ket{1},
\]
all mask bits become $\ket{+}$ states, and we initialize the encoding-output qubits to $\ket{0}$. The circuit for computing the randomized encoding then consists of a sequence of CNOT gates. Tracing out the source registers (i.e. the $X_j$ qubits) leaves a diagonal density matrix whose diagonal is the classical distribution of the randomized encoding. Therefore the entanglement entropy across the source-output cut is the Shannon entropy of that encoding.

The construction is 2D-local because it overlays two one-dimensional dependency patterns in a row-and-column layout, using a constant-size plaquette for each matrix position $(i,j)$. For each fixed column $j$, the bits $\bw_{j,0},\ldots,\bw_{j,q-1}$ are computed by a vertical nearest-neighbor chain as the row index $i$ varies. For each fixed row $i$, the bits $\bz_{i,1},\ldots,\bz_{i,m}$ are computed by horizontal nearest-neighbor plaquettes as the column index $j$ varies. The low and high branches differ only in which local CNOT from $\br_{j,i}$ to $\bz_{i,j}$ is included, according to the public matrix entry $\mathbf{N}_{i,j}^b$. We give a high-level illustration of the circuit in~\cref{fig:2d-coherent-re}, where the gray highlight marks one local plaquette; \cref{fig:local-plaquette} expands that plaquette into a planar constant-size patch.

For the Hamiltonian applications, the 2D result follows by the
conjugated-projector construction $H_{\calC}=\sum_{i=1}^n \calC\ket{1}\!\bra{1}_i\calC^\dagger,$ described in the previous subsection. For the 1D result we
use the quantum Turing Machine (QTM)-to-Hamiltonian method of
Bouland--Zhang--Zhou~\cite{groundstatehardness}. A reversible QTM simulates
the preparation circuit, moves its output to a fixed output track, and then
enters a long clock-only padding phase. Throughout that phase the output
track is exactly the desired state, while every changing auxiliary track is
independent of the output. The resulting late-history state therefore
factorizes across the fixed physical output track. Taking the padding to be long enough makes the complete history state close to this factorized late
history state, so entropy continuity preserves the fixed-cut gap. The
QTM-to-Hamiltonian construction is 1D local, frustration-free, has a unique
history-state ground state, and is inverse-polynomially gapped.

For the formal details, \Cref{sec:prelim} collects the information-theoretic and cryptographic preliminaries, \Cref{sec:retheorem} proves the randomized-encoding theorem, \Cref{sec:qubitproof} proves the pseudoentanglement theorem, \Cref{sec:hamapp} proves the Hamiltonian consequences, \appref{app:pmghz} gives an alternative construction of pseudoentangled states from \emph{poor man's GHZ states}, and \appref{app:crh} gives a construction of a constant-local collision-resistant hash based on the random-code binary shortest-vector problems $\mathsf{bSVP}$.

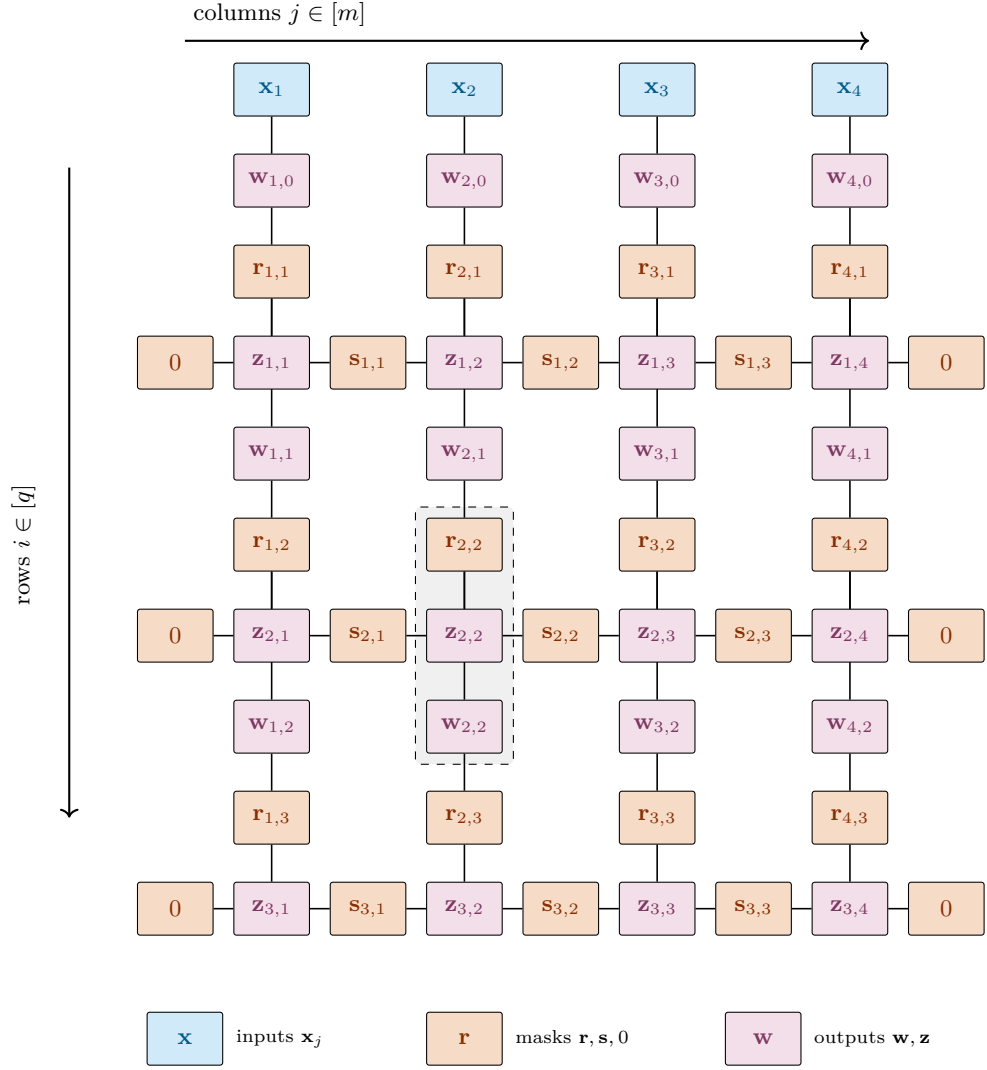
\begin{figure}[!ht]
\centering
\captionsetup{font=footnotesize,skip=3pt}

\definecolor{oiSkyBlue}{RGB}{86,180,233}
\definecolor{oiSkyBlueText}{RGB}{0,90,135}
\definecolor{oiVermillion}{RGB}{213,94,0}
\definecolor{oiVermillionText}{RGB}{135,45,0}
\definecolor{oiPurple}{RGB}{204,121,167}
\definecolor{oiPurpleText}{RGB}{125,55,98}

\scalebox{0.80}{%
\begin{tikzpicture}[
    x=2.55cm, y=0.86cm,
    box/.style={rectangle, rounded corners=1pt, draw, minimum width=10mm,
        minimum height=7mm, inner sep=1pt, font=\small},
    inbox/.style={box, fill=oiSkyBlue!28, text=oiSkyBlueText},
    randr/.style={box, fill=oiVermillion!22, text=oiVermillionText},
    rands/.style={box, fill=oiVermillion!22, text=oiVermillionText},
    outw/.style={box, fill=oiPurple!24, text=oiPurpleText},
    outz/.style={box, fill=oiPurple!24, text=oiPurpleText},
    vwire/.style={-,semithick,black},
    hwire/.style={-,semithick,black},
    mwire/.style={-,semithick,black},
    plaquette/.style={draw=black, dashed, fill=black!6, rounded corners=2pt,
        inner sep=4pt},
]

\def\m{4}\def\q{3}

\def\rowstep{4.2}
\def\topgap{0.0}
\newcommand{\Yr}[1]{{-\rowstep*#1-\topgap}}
\newcommand{\Ys}[1]{{-\rowstep*#1-\topgap-1.4}}
\newcommand{\Yw}[1]{{-\rowstep*#1-\topgap-2.8}}

\foreach \j in {1,...,\m}{
  \foreach \i in {1,...,\q}{
    \node[randr] (r\j\i) at (\j, \Yr{\i}) {$\mathbf{r}_{\j,\i}$};
  }
}

\foreach \j in {1,...,\m}{
  \node[inbox] (x\j) at (\j, {-\rowstep+2.8}) {$\mathbf{x}_{\j}$};
}

\foreach \i in {1,...,\q}{
  \node[rands] (s\i0) at (0.5, \Ys{\i}) {$0$};
  \foreach \j in {1,...,\m}{
    \ifnum\j<\m
      \node[rands] (s\i\j) at (\j+0.5, \Ys{\i}) {$\mathbf{s}_{\i,\j}$};
    \else
      \node[rands] (s\i\j) at (\j+0.5, \Ys{\i}) {$0$};
    \fi
  }
}

\foreach \j in {1,...,\m}{
  \node[outw] (w\j0) at (\j, {-\rowstep+1.4}) {$\mathbf{w}_{\j,0}$};
  \draw[vwire] (x\j) -- (w\j0);
  \draw[vwire] (w\j0) -- (r\j1);
  \foreach \i in {1,...,\q}{
    \ifnum\i<\q
      \pgfmathtruncatemacro{\ip}{\i+1}
      \node[outw] (w\j\i) at (\j, \Yw{\i}) {$\mathbf{w}_{\j,\i}$};
      \draw[vwire] (r\j\i) -- (w\j\i);
      \draw[vwire] (w\j\i) -- (r\j\ip);
    \fi
  }
}

\foreach \i in {1,...,\q}{
  \foreach \j in {1,...,\m}{
    \pgfmathtruncatemacro{\jm}{\j-1}
    \node[outz] (z\i\j) at (\j, \Ys{\i}) {$\mathbf{z}_{\i,\j}$};
    \draw[hwire] (s\i\jm) -- (z\i\j);
    \draw[hwire] (z\i\j) -- (s\i\j);
    \draw[mwire] (r\j\i) -- (z\i\j);
  }
}

\begin{scope}[on background layer]
  \node[plaquette, fit=(r22)(z22)(w22)] {};
\end{scope}

\draw[->,thick] (0.55,{-\rowstep+3.55}) -- (4.1,{-\rowstep+3.55});
\node[font=\small,anchor=south west] at (0.55,{-\rowstep+3.65}) {columns $j\in[m]$};
\draw[->,thick] (-0.05,{-\rowstep+1.6}) -- (-0.05,\Yr{3});
\node[font=\small,rotate=90,anchor=south] at (-0.18,\Yr{2}) {rows $i\in[q]$};

\begin{scope}[shift={(0.55,-16.0)}]
  \node[inbox] at (0,0) {$\mathbf{x}$};
  \node[anchor=west,font=\scriptsize] at (0.22,0)
    {inputs $\mathbf{x}_j$};
  \node[randr] at (1.45,0) {$\mathbf{r}$};
  \node[anchor=west,font=\scriptsize] at (1.67,0)
    {masks $\mathbf{r},\mathbf{s},0$};
  \node[outw] at (3.00,0) {$\mathbf{w}$};
  \node[anchor=west,font=\scriptsize] at (3.22,0)
    {outputs $\mathbf{w},\mathbf{z}$};
\end{scope}

\end{tikzpicture}%
}

\caption{
Schematic 2D layout of the randomized encoding with
four input columns and three representative rows. Sky-blue boxes are input
bits, vermillion boxes are mask bits, and purple boxes are
outputs. The
$\bw$-outputs are computed by column-local nearest-neighbor gadgets:
$\bw_{j,0}$ depends only on $\bx_j$ and $\br_{j,1}$, while $\bw_{j,1}$
and $\bw_{j,2}$ depend only on adjacent $\br$-masks in the same column.
The $\bz$-outputs are computed inside constant-size plaquettes:
$\bz_{i,j}$ depends only on the two neighboring horizontal masks
$\bs_{i,j-1}$ and $\bs_{i,j}$, and, when $\mathbf{N}_{i,j}=1$, the local
vertical mask $\br_{j,i}$. The vertical wires through the $\bw$-boxes should
be read as connections to the adjacent $\br$-masks above and below them, not
as connections to the $\bz$-boxes in the horizontal lane; each $\br_{j,i}$
feeds both the local $\bz_{i,j}$ box and the neighboring $\bw$ boxes in its
column. The gray region highlights one constant-size plaquette: in the
actual 2D circuit, this local patch may contain a constant number of ancillas
so that the variables $\br_{j,i}$, $\bz_{i,j}$, and the neighboring $\bw$
boxes can be coupled within the plaquette using only nearest-neighbor gates.
Each horizontal mask $\bs_{i,j}$ is shared by the two adjacent plaquettes in
row $i$, with boundary masks $\bs_{i,0}=\bs_{i,m}=0$.
}
\label{fig:2d-coherent-re}
\end{figure}

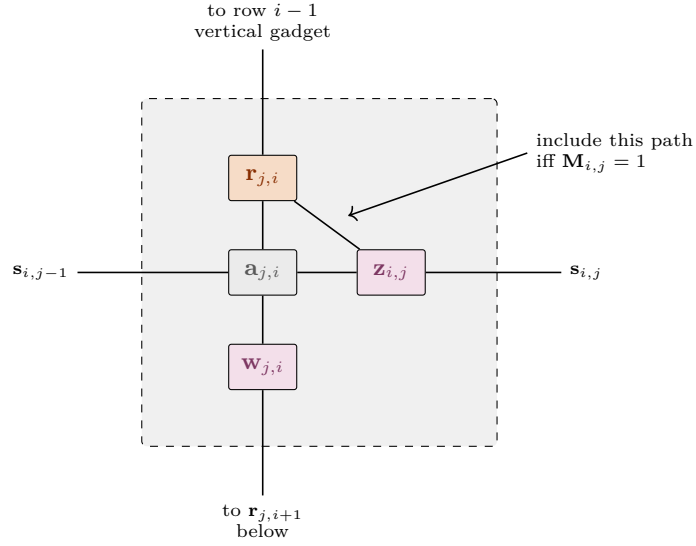
\begin{figure}[t]
\centering

\definecolor{oiVermillion}{RGB}{213,94,0}
\definecolor{oiVermillionText}{RGB}{135,45,0}
\definecolor{oiPurple}{RGB}{204,121,167}
\definecolor{oiPurpleText}{RGB}{125,55,98}

\begin{tikzpicture}[
    x=1cm, y=1cm,
    box/.style={rectangle, rounded corners=1pt, draw, minimum width=9mm,
        minimum height=6mm, inner sep=1pt, font=\small},
    mask/.style={box, fill=oiVermillion!22, text=oiVermillionText},
    output/.style={box, fill=oiPurple!24, text=oiPurpleText},
    anc/.style={box, fill=black!8, text=black!65},
    wire/.style={-, semithick, black},
    note/.style={->, semithick, black},
    plaquette/.style={draw=black, dashed, fill=black!6, rounded corners=2pt}
]

\node[plaquette, minimum width=4.7cm, minimum height=4.6cm] (P) at (0,0) {};

\coordinate (topout) at (-0.75,2.95);
\coordinate (topin) at (-0.75,2.30);
\coordinate (botin) at (-0.75,-2.30);
\coordinate (botout) at (-0.75,-2.95);
\coordinate (leftout) at (-3.20,0);
\coordinate (leftin) at (-2.85,0);
\coordinate (rightin) at (2.85,0);
\coordinate (rightout) at (3.20,0);

\node[mask] (r) at (-0.75,1.25) {$\mathbf{r}_{j,i}$};
\node[output] (w) at (-0.75,-1.25) {$\mathbf{w}_{j,i}$};
\node[output] (z) at (0.95,0) {$\mathbf{z}_{i,j}$};
\node[anc] (ax) at (-0.75,0) {$\mathbf{a}_{j,i}$};

\draw[wire] (topout) -- (topin) -- (r);
\draw[wire] (w) -- (botin) -- (botout);
\draw[wire] (leftout) -- (leftin) -- (ax) -- (z);
\draw[wire] (z) -- (rightin) -- (rightout);

\draw[wire] (r) -- (ax) -- (w);
\draw[wire] (r) -- (z);

\node[font=\scriptsize, align=center, anchor=south] at (topout)
  {to row $i-1$\\ vertical gadget};
\node[font=\scriptsize, align=center, anchor=north] at (botout)
  {to $\mathbf{r}_{j,i+1}$\\ below};
\node[font=\scriptsize, anchor=east] at (leftout)
  {$\mathbf{s}_{i,j-1}$};
\node[font=\scriptsize, anchor=west] at (rightout)
  {$\mathbf{s}_{i,j}$};
\node[font=\scriptsize, align=left, anchor=west] (pathnote) at (2.75,1.58)
  {include this path\\ iff $\mathbf{N}_{i,j}=1$};
\draw[note] (pathnote.west) -- (0.38,0.75);

\end{tikzpicture}

\caption{
Local implementation of one plaquette. The four external ports carry
$\br_{j,i}$ from above, $\br_{j,i+1}$ to the row below, and the two
horizontal masks $\bs_{i,j-1}$ and $\bs_{i,j}$. The plaquette computes
$\bw_{j,i}=\br_{j,i}\oplus\br_{j,i+1}$ and
$\bz_{i,j}=\bs_{i,j-1}\oplus\bs_{i,j}\oplus
\mathbf{N}_{i,j}\br_{j,i}$.
The gray register $\mathbf{a}_{j,i}$ is only local workspace, initialized to
$0$ and uncomputed to $0$ after use. One possible local schedule is as
follows. First copy $\br_{j,i}$ into the workspace, apply it to
$\bw_{j,i}$, and uncompute the workspace:
$\br_{j,i}\to\mathbf{a}_{j,i}\to\bw_{j,i}$ followed by
$\br_{j,i}\to\mathbf{a}_{j,i}$. Then apply the incoming
$\br_{j,i+1}$ from the bottom port to $\bw_{j,i}$. Next reuse the same
workspace for the horizontal mask: copy $\bs_{i,j-1}$ into
$\mathbf{a}_{j,i}$, apply it to $\bz_{i,j}$, and uncompute it; then apply
$\bs_{i,j}$ directly to $\bz_{i,j}$. Finally, include the direct
$\br_{j,i}\to\bz_{i,j}$ CNOT exactly when $\mathbf{N}_{i,j}=1$.
}
\label{fig:local-plaquette}
\end{figure}

\clearpage

\subsection{Related work}
\label{subsec:related}

The notion of pseudoentanglement was introduced by Aaronson~\etal~\cite{ITCS:ABFGVZ24} taking inspiration from~\cite{gheorghiu}. Their construction is based on quantum pseudorandom states and yields states with entropy on the order of $\log n$ across every cut, while remaining computationally indistinguishable from states with near-maximal entanglement entropy. Our work is complementary to theirs. We obtain a weaker fixed-cut guarantee, but in a setting that is explicitly \emph{not} pseudorandom: the preparing circuits are public and 2D-local constant-depth. In view of the learnability result of Huang \etal~\cite{shallowlearning}, this gives a sharp separation between pseudorandomness and pseudoentanglement. Another recent direction, due to Cheng, Feng, and Ippoliti~\cite{ChengFI25}, uses tensor-network constructions to realize more flexible pseudoentanglement patterns, including holographic examples; our emphasis here is instead on public-key hardness from post-quantum assumptions and on low-depth local circuit preparation.

A direct precursor to the present paper is the work of Gheorghiu and Hoban~\cite{gheorghiu}, which showed hardness of entropy estimation for log-depth circuits and for constant-depth circuits with bounded fan-in and unbounded fan-out gates, under \textsf{LWE}. That construction also used randomized encodings, but while those encodings could be computed in constant depth, they used unbounded fan-out gates, whereas our result here has bounded fan-out. It is unclear whether the bounded fan-out randomized encoding we used could also be applied to non-binary linear maps, allowing us to base hardness on \textsf{LWE} rather than Dense-Sparse \textsf{LPN}.

Bouland \etal~\cite{publickey} introduced public-key pseudoentanglement and used it to prove hardness of learning ground-state entanglement structure from Hamiltonian descriptions. Their construction established that pseudoentanglement can be useful in the computational-complexity model, where the input is a classical description of a circuit or Hamiltonian rather than copies of a state. Our main theorem can be viewed as a shallow-circuit analogue of this line of work: we retain the public-key flavor, but now with 2D-local constant-depth state preparation. Closely related in motivation is the recent work of Akers \etal~\cite{AkersBCKMV24}, which uses a pseudoentanglement-style hiding phenomenon to argue for the potential hardness of geometry reconstruction in the AdS/CFT dictionary.

The recent work of Bouland, Zhang, and Zhou~\cite{groundstatehardness} strengthened the Hamiltonian application by proving hardness for geometrically local Hamiltonians with a near-area-law versus near-volume-law promise. Their 1D result is based on factoring-type assumptions, while their 2D result is based on \textsf{LWE}. Our Hamiltonian consequences are weaker in the geometric promise, since they concern only a specified cut inherited from the circuit construction. However, the 2D consequence has constant gap, and the 1D consequence gives post-quantum fixed-cut hardness from a code-based assumption.

There has also been rapid recent progress on low-depth pseudorandom unitaries. Schuster, Haferkamp, and Huang~\cite{SchusterHH24} construct PRUs in depth $\polylog(n)$ in one-dimensional geometries and depth $\poly(\log\log n)$ in all-to-all architectures using ordinary bounded-fan-in unitary circuits. Foxman, Parham, Vasconcelos, and Yuen~\cite{FoxmanPVY26} show that PRUs can be realized in constant-time models such as \QACzf and equivalent intermediate-measurement/feedforward models. It is known, however, that both PRSs and PRUs cannot be achieved in \QNCz, as those circuits are efficiently learnable~\cite{shallowlearning}.

Our randomized-encoding construction sits in the broader low-complexity cryptography program where randomized encodings are used to reduce locality and obtain primitives such as low-depth, local one-way functions, pseudorandom generators and simple hash functions~\cite{Applebaum14,JC:AppIshKus18,STOC:Applebaum12,FOCS:Applebaum17,EPRINT:Applebaum17}. To that end, in \appref{app:crh} we give, as a corollary of \Cref{thm:reintro}, a constant-local collision-resistant hash function based on a random-code variant of $\mathsf{bSVP}$. Previous constructions of such functions were based on \textsf{bSVP} for sparse random matrices (i.e. LDPC codes), rather than dense matrices. The catch, however, is that our hash function achieves much weaker compression compared to previous works.

\section{Discussion and open questions}
\label{sec:discussion}

As mentioned, our construction of constant-depth pseudoentangled states gives a separation between pseudoentanglement and pseudorandomness in the shallow-circuit regime. Pseudorandomness is impossible for \QNCz state-preparation circuits because such circuits are efficiently learnable, whereas we show that pseudoentanglement can still be achieved. The key mechanism is a constant-locality randomized encoding for binary linear maps, which may also be useful in classical cryptographic settings. However, several limitations and open questions remain.

The main limitation is that the hidden cut in our construction is fixed and non-geometric. It is determined by the input-output structure of the randomized encoding rather than by a connected spatial region of the lattice. As a result, the construction cannot be used directly to prove hardness for learning entanglement structure in more physically motivated settings, such as the geometric cuts that arise in recent quantum-gravity-inspired work. An interesting open problem is whether the ideas used here can be adapted to give constant-depth constructions of states with ``geometric pseudoentanglement,'' closer to the setting of~\cite{groundstatehardness}.

There are also general limitations on what such a construction could achieve. In constant depth, one can compute the entropy of any subset of $O(\log n)$ output qubits in polynomial time. Indeed, the reduced density matrix on such a subset has dimension $2^{O(\log n)}=\poly(n)$, so it can be written down explicitly, and a light-cone argument shows that for a depth-$d$ bounded-fan-in circuit it depends on only $O(2^d\log n)$ input qubits, which is still $O(\log n)$ for constant $d$. Moreover, if the underlying assumption, such as Dense-Sparse \textsf{LPN}, is assumed to be secure against $2^{o(n)}$-time algorithms, then one cannot expect hardness for cuts of size $o(n)$. Effectively, fixed-depth hardness should only be possible for cuts of linear size. It would therefore be interesting to find a construction that is tight in this sense: one for which computing the von Neumann entropy is hard for all linear-size cuts.

Another direction is to look for applications beyond entropy estimation. Schuster~\etal recently showed that recognizing phases of matter is intractable by leveraging low-depth PRUs~\cite{SchusterKYH26}. Their result does not apply to constant-locality Hamiltonians, because PRUs need a growing light cone in order to hide information. Pseudoentangled states avoid this particular obstruction: they hide only an entanglement property, not the entire state. This raises the question of whether pseudoentanglement could help prove hardness of recognizing phases of matter for constant-locality Hamiltonians. There is also a related possibility. Since constant-depth circuits do not change the phase of a state, our circuit-output states are all trivial in the usual condensed-matter sense. Nevertheless, the construction might still be useful for hiding witnesses of phase equivalence, and hence for showing that deciding whether two succinctly described states or Hamiltonians lie in the same phase can be computationally hard.

Our proof appears to be specific to von Neumann entropy, or equivalently to Shannon entropy after the source register is traced out. The Dense-Sparse \textsf{LPN} analysis compares the two branches only on the typical sparse event $|X|\le 2t$, whose complement has negligible probability. This tail is harmless for Shannon entropy, however for other R\'enyi orders, negligible-probability events can matter in qualitatively different ways. For example, the max-entropy $H_0$, which is the logarithm of the support size, depends on the full support of the distribution, so rare nonsparse inputs can make the low branch have large rank entropy even though it is lossy on the sparse set. Conversely, large R\'enyi orders, including min-entropy $H_{\infty}$, are controlled by the largest point probabilities rather than by the typical support size. Thus the injective-versus-lossy guarantee supplied by the lossy-function construction does not by itself give a robust R\'enyi entropy gap. Extending the result to other entropies would likely require a source distribution or cryptographic primitive that controls the entire output spectrum, not just its Shannon entropy on a typical set.

Finally, the construction is not strictly over a fixed finite gate set: the initial single-qubit rotations depend on the parameter $t/n$. It would be interesting to know whether our results can be recovered by constant-depth circuits built from a fixed, constant-size gate set. Additionally, it is also worth determining whether such a construction is possible with 1D local circuits (even when allowing arbitrarily parametrized gates).

\subsection*{Acknowledgements and AI statement}
I am grateful to my good friend Matty Hoban for many insightful discussions on this topic.
I have made extensive use of OpenAI's Codex, Anthropic's Claude Code as well as ChatGPT 5.5 Extended Pro in writing this manuscript. The original idea for constant depth pseudoentanglement was mine; I specifically came up with the construction in~\cref{app:pmghz} as well as the applications to hardness of learning entanglement structure of local Hamiltonian ground states. However, after prompting ChatGPT 5.5 about whether the construction could be simplified, it came up with the randomized encoding scheme that became the main result, as it is simpler than the Poor man's GHZ state construction in~\cref{app:pmghz} and also has an independent cryptographic application. I then used the AI models to aid in writing all sections of the paper and in checking correctness of the results. I have also independently re-derived and checked the proofs myself.

\section{Preliminaries}
\label{sec:prelim}

We follow the notation from~\cite{nielsenchuang}. All Hilbert spaces in this paper are finite-dimensional over $\C$. Throughout, $[m] \defeq \{1,\ldots,m\}$, and all vector and matrix additions and products are over $\F_2$ unless explicitly stated otherwise. Bold lowercase letters such as $\bx,\by,\bz,\bw$ denote vectors or bit strings, while bold uppercase letters such as $\bA,\bC,\bS,\bE$ denote matrices. For a bit string $\bx \in \bits^m$, the quantity $|\bx|$ denotes its Hamming weight; for a set $S$, the same notation $|S|$ denotes its cardinality. We write $X \sim \Bern(p)$ for a Bernoulli random variable with $\prob{X=1}=p$ and $\prob{X=0}=1-p$, and $X \sim \Bern(p)^{\otimes m}$ for the product distribution of $m$ independent $\Bern(p)$ coordinates. Note that for random variables and sets we uses non-bolded capital letters. Finally, $\Unif(S)$ denotes the uniform distribution over a finite set $S$.

\subsection{Classical and quantum information}
\label{subsec:qi}

For a classical random variable $X$ distributed over a finite set $\Omega$, we write
\[
    H(X) = - \sum_{x \in \Omega} \prob{X=x} \log_2 \prob{X=x}
\]
for its Shannon entropy. We also write
\[
    H_\alpha(X)
    =
    \frac{1}{1-\alpha}\log_2 \sum_{x \in \Omega} \prob{X=x}^\alpha
\]
for its classical R\'enyi entropy when $\alpha \in (0,\infty)\setminus\{1\}$, with the standard conventions
\[
    H_0(X)=\log_2 |\mathrm{supp}(X)|
    \qquad\text{and}\qquad
    H_\infty(X)=-\log_2 \max_{x \in \Omega} \prob{X=x}.
\]
We also write
\[
    h(p) = -p \log_2 p - (1-p)\log_2(1-p)
\]
for the binary entropy function. For a density matrix $\rho$, we write
\[
    \Se{\rho} = -\Tr{\rho \log_2 \rho}
\]
for its von Neumann entropy. For a bipartite pure state $\ket{\psi}_{AB}$, the entanglement entropy across the cut $A:B$ is
\[
    \Se{\Ptr{A}{\ket{\psi}\bra{\psi}}}
    =
    \Se{\Ptr{B}{\ket{\psi}\bra{\psi}}}.
\]
For $\alpha \in (0,\infty)\setminus\{1\}$, we also write
\[
    \mathsf{S}_\alpha(\rho)
    =
    \frac{1}{1-\alpha}\log_2 \Tr{\rho^\alpha}
\]
for the quantum R\'enyi entropy of $\rho$. We use the standard conventions $\mathsf{S}_0(\rho)=\log_2 \mathrm{rank}(\rho)$ and $\mathsf{S}_\infty(\rho)=-\log_2 \|\rho\|_\infty$, and the fact that $\mathsf{S}_\alpha(\rho)$ is nonincreasing in $\alpha$. Both the von Neumann entropy and the R\'enyi entropies are invariant under local unitaries on either side of the cut.

\begin{proposition}[Classical coherent states]
    \label{prop:coherententropy}
    Let $X$ be a classical random variable over a finite set $\Omega$, let $f : \Omega \to \bits^r$, and define
    \[
        \ket{\Phi_f} = \sum_{x \in \Omega} \sqrt{\prob{X=x}} \ket{x}\ket{f(x)}.
    \]
    Then
    \[
        \Ptr{\mathrm{first}}{\ket{\Phi_f}\bra{\Phi_f}}
        =
        \sum_{y \in \bits^r} \prob{f(X)=y}\ket{y}\bra{y}.
    \]
    Consequently, for every $\alpha \in [0,\infty]$,
    \[
        \mathsf{S}_\alpha\!\left(\Ptr{\mathrm{first}}{\ket{\Phi_f}\bra{\Phi_f}}\right)
        =
        H_\alpha(f(X)).
    \]
    In particular, the entanglement entropy of $\ket{\Phi_f}$ across the cut separating the two registers is exactly $H(f(X))$.
\end{proposition}

\begin{proof}
    Since the computational basis states in the first register are orthonormal, tracing out that register removes all off-diagonal terms with different first-register labels. The reduced density matrix on the second register is therefore diagonal in the computational basis, with diagonal entries equal to the probabilities of the random variable $f(X)$. Its quantum R\'enyi entropies are therefore exactly the corresponding classical R\'enyi entropies of $f(X)$, and the $\alpha=1$ case is the stated von Neumann/Shannon identity.
\end{proof}

\begin{proposition}[Chernoff bound]
    \label{prop:chernoff}
    Let $X_1,\ldots,X_m$ be independent Bernoulli random variables and let $S = \sum_{i=1}^m X_i$ with expectation $\mu = \mathbb{E}[S]$. Then
    \[
        \prob{S \geq 2\mu} \leq \exp(-\mu/3).
    \]
\end{proposition}

\subsection{Local Hamiltonians}
\label{subsec:hamiltonians}

A 2D $k$-local Hamiltonian on $n$ qubits is an operator of the form
\[
    H=\sum_j h_j
\]
where the qubits are embedded in a rectangular subset of the square lattice and each term $h_j$ acts on at most $k$ qubits whose lattice diameter is at most $k$. A 1D $k$-local Hamiltonian is defined analogously for qubits embedded on a line: each term acts on at most $k$ qubits whose diameter along the line is at most $k$.

We say that $H$ is frustration-free if every $h_j$ is positive semidefinite and the ground energy of $H$ is $0$. The spectral gap of $H$ is
\[
    \Delta(H)=\lambda_1(H)-\lambda_0(H),
\]
where $\lambda_0(H)\le \lambda_1(H)\le \cdots$ are the eigenvalues of $H$ in nondecreasing order.

For a subset $A$ of lattice sites, write
\[
    \mathrm{Area}(A)
\]
for the number of nearest-neighbor lattice edges with one endpoint in $A$ and the other in $\overline{A}$. Following~\cite{groundstatehardness}, one may view a state as near-area-law entangled across $A$ if its entanglement entropy is $\widetilde{O}(\mathrm{Area}(A))$, and as near-volume-law entangled across $A$ if its entanglement entropy is $\widetilde{\Omega}(\min\{|A|,|\overline{A}|\})$.

\begin{definition}[Fixed-cut ground-state entanglement learning]
    \label{def:lghes}
    Let $a,b:\N\to\mathbb{R}_{\ge 0}$ satisfy $a(n)<b(n)$ for all sufficiently large $n$. The problem $\mathsf{LGHES}_{a,b}$ is the following promise problem:
    \begin{itemize}
        \item \bfm{Input.} A classical description of a geometrically local Hamiltonian $H$ on $n$ qubits, together with a subset $W\subseteq[n]$.
        \item \bfm{Promise.} The Hamiltonian $H$ has a unique ground state $\ket{\phi_H}$, and either
        \[
            \Se{\Ptr{W}{\ket{\phi_H}\bra{\phi_H}}}\le a(n)
        \]
        or
        \[
            \Se{\Ptr{W}{\ket{\phi_H}\bra{\phi_H}}}\ge b(n).
        \]
        \item \bfm{Task.} Decide which of the two cases holds.
    \end{itemize}
    This is the fixed-cut version of the ground-state entanglement learning task studied in~\cite{groundstatehardness}. The latter work focuses on the stronger near-area-law versus near-volume-law promise for geometrically local cuts.
\end{definition}

\subsection{Cryptography}
\label{subsec:crypto}

A function $\mu : \N \to \mathbb{R}_{\geq 0}$ is \emph{negligible} if for every constant $c > 0$ there exists $N_c$ such that $\mu(\lambda) \leq \lambda^{-c}$ for all $\lambda \geq N_c$. We write $\negl(\lambda)$ for an unspecified negligible function of the security parameter $\lambda$.

\begin{definition}[Fan-in and fan-out]
    \label{def:faninout}
    Let $F:\bits^a\to\bits^b$ be a Boolean map, and write
    \[
        F=(F_1,\ldots,F_b).
    \]
    The \emph{fan-in} of the output coordinate $F_i$ is the number of input coordinates on which $F_i$ depends. The \emph{fan-out} of an input coordinate is the number of output coordinates that depend on it. We say that $F$ has fan-in at most $k$ and fan-out at most $d$ if every output coordinate has fan-in at most $k$ and every input coordinate has fan-out at most $d$. These quantities are \emph{bounded} if $k,d=O(1)$, independent of the input length, and \emph{unbounded} if they are allowed to grow with the input length. For circuit descriptions, we also use the standard gate-level meaning: fan-in is the number of incoming wires to a gate, and fan-out is the number of outgoing wires or targets to which a gate output is sent; this is illustrated for a gate $G$ in~\cref{fig:fanin-fanout}.
\end{definition}

\begin{figure}[h!]
\centering
\begin{tikzpicture}[
    x=1cm,y=1cm,
    wire/.style={-,semithick,black},
    gate/.style={rectangle,draw,fill=black!6,minimum width=1.5cm,
        minimum height=1.0cm,font=\small},
    brace/.style={decorate,decoration={brace,amplitude=4pt},semithick}
]
    \node[gate] (g) at (0,0) {$G$};

    \foreach \y/\lab in {0.75/1,0.35/2,-0.75/k}{
        \draw[wire] (-2.6,\y) -- (g.west);
        \node[font=\scriptsize,anchor=east] at (-2.7,\y) {$x_{\lab}$};
    }
    \node[font=\scriptsize] at (-2.6,-0.18) {$\vdots$};

    \foreach \y/\lab in {0.55/1,0.18/2,-0.55/d}{
        \draw[wire] (g.east) -- (2.6,\y);
        \node[font=\scriptsize,anchor=west] at (2.7,\y) {$y_{\lab}$};
    }
    \node[font=\scriptsize] at (2.6,-0.18) {$\vdots$};

    \draw[brace] (-3.35,-0.9) -- (-3.35,0.9)
        node[midway,left=6pt,font=\scriptsize] {fan-in};
    \draw[brace] (3.35,0.7) -- (3.35,-0.7)
        node[midway,right=6pt,font=\scriptsize] {fan-out};
\end{tikzpicture}
\caption{
Gate-level fan-in and fan-out. In the example shown, the gate has fan-in $k$
and fan-out $d$. These quantities are bounded when $k,d=O(1)$, independent of
the total problem size.
}
\label{fig:fanin-fanout}
\end{figure}
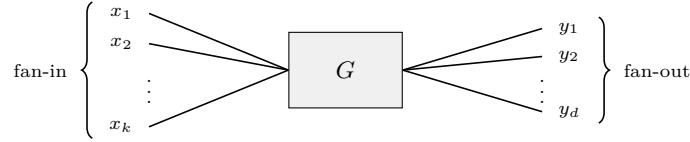

\begin{definition}[Randomized encoding]
    \label{def:randenc}
    Let $f:\bits^m\to\bits^\ell$. A randomized encoding of $f$ consists of an encoding map
    \[
        \mathsf{Enc}:\bits^m\times\bits^r\to\bits^s
    \]
    together with algorithms $\mathsf{Dec}$ and $\mathsf{Sim}$ such that:
    \begin{enumerate}
        \item \bfm{Correctness.} For every $\bx\in\bits^m$ and every randomness string $\rho\in\bits^r$,
        \[
            \mathsf{Dec}(\mathsf{Enc}(\bx;\rho))=f(\bx).
        \]
        \item \bfm{Privacy.} If $R\leftarrow\bits^r$ is uniform, then the distribution $\mathsf{Enc}(\bx;R)$ can be simulated from $f(\bx)$ alone. Concretely, $\mathsf{Sim}(f(\bx))$ is required to be computationally, statistically, or perfectly indistinguishable from $\mathsf{Enc}(\bx;R)$, depending on the notion under consideration.
    \end{enumerate}
    We say the randomized encoding is \emph{perfect} if these two distributions are exactly identical for every input $\bx$.
\end{definition}

\begin{definition}[Dense-Sparse \textsf{LPN} public-key distributions]
    \label{def:densesparselpn}
    Fix an even dimension $d$, a number of samples $m$, a constant column
    weight $k\ge 3$, and a number $\ell$ of samples.
    Sample
    \[
        \mathbf{M}\leftarrow\mathcal{M}_{\mathrm{sp}}
            \subseteq\F_2^{d\times m},
        \qquad
        \mathbf{T}\leftarrow\Unif(\F_2^{d/2\times d}),
        \qquad
        \bA=\mathbf{T}\mathbf{M},
    \]
    where every column of $\mathbf{M}$ has weight $k$ and
    $\mathcal{M}_{\mathrm{sp}}$ is a good sparse-matrix distribution in the
    sense of Dao--Jain~\cite{dao2024lossy}.  Independently sample
    \[
        \bS\leftarrow\Unif(\F_2^{\ell\times d/2}),
        \qquad
        \bE\leftarrow\Bern(\varepsilon)^{\ell\times m}.
    \]
    The Dense-Sparse \textsf{LPN} distribution consists of
    $(\bA,\bS\bA\oplus\bE)$; its uniform counterpart is
    $(\bA,\mathbf{U})$ for uniform
    $\mathbf{U}\in\F_2^{\ell\times m}$.
\end{definition}

\begin{conjecture}[Dense-Sparse \textsf{LPN} assumption]
    \label{conj:densesparselpn}
    For every constant $D>1$, there exist an even constant $k\ge 6$ and
    $\rho\in(0,1)$ such that
    \[
        a=1+(k/2-1)(1-\rho)
        \qquad\text{satisfies}\qquad
        1<a<k/4.
    \]
    The restriction on $k$ and $a$ is the range in which the good
    sparse-matrix distributions used by Dao--Jain are known.
    Set
    \[
        m=d^a,
        \qquad
        \delta=\frac{Dk-a}{Dk-1},
        \qquad
        s=d^\delta.
    \]
    Since $a<k/4<k/2$, we have $m<d^{k/2}$ and this value of $\delta$
    lies in the admissible interval for the formal Dao--Jain construction:
    \[
        1-\frac{k/2-1}{Dk-1}<\delta<1.
    \]
    Equivalently, before the usual integer rounding,
    \[
        m=d^{1+(Dk-1)(1-\delta)},
        \qquad
        m s^{Dk-1}=d^{Dk}.
    \]
    For a sufficiently small constant $c>0$ and a sufficiently large
    constant $C>0$, take
    \[
        \varepsilon=\frac{c}{s}
        =\Theta(d^{-\delta})
        =\Theta\!\left(
            \left(\frac{m}{d^{Dk}}\right)^{1/(Dk-1)}
        \right),
        \qquad
        \ell=C s\log(m/s).
    \]
    Take ceilings for $m$, $s$, and $\ell$ when integer parameters are
    required.
    These are the formal lossy-trapdoor-function parameter choices of
    Dao--Jain~\cite[Fig.~3]{dao2024lossy}, with the same compression
    parameter $D$.  We assume that, for these parameters,
    polynomially many parallel Dense-Sparse \textsf{LPN} samples are
    computationally indistinguishable from uniform for QPT adversaries:
    \[
        (\bA,\bS\bA\oplus\bE)
        \approx_c
        (\bA,\mathbf{U}).
    \]
\end{conjecture}

\begin{definition}[Sparse-input lossy trapdoor functions]
    \label{def:ltdf}
    Fix a family of input sets $T_\lambda \subseteq \bits^{m(\lambda)}$ and a constant $\Gamma > 1$. A sparse-input lossy trapdoor function family with lossiness factor $\Gamma$ consists of polynomial-time algorithms
    \[
        \mathsf{Gen}^{high}(1^\lambda)\to(pk,td),
        \qquad
        \mathsf{Gen}^{low}(1^\lambda)\to pk,
        \qquad
        \mathsf{Eval}(pk,\cdot),
    \]
    such that, except with negligible probability over key generation, the following hold:
    \begin{enumerate}
        \item \bfm{Injective mode.} If $(pk,td)\leftarrow \mathsf{Gen}^{high}(1^\lambda)$, then $\mathsf{Eval}(pk,\cdot)$ is injective on $T_\lambda$, and the trapdoor $td$ allows efficient inversion on $\mathsf{Eval}(pk,T_\lambda)$.
        \item \bfm{Lossy mode.} If $pk \leftarrow \mathsf{Gen}^{low}(1^\lambda)$, then
        \[
            |\mathsf{Eval}(pk,T_\lambda)| \leq |T_\lambda|^{1/\Gamma}.
        \]
        \item \bfm{Indistinguishability.} The public keys produced by $\mathsf{Gen}^{high}$ and $\mathsf{Gen}^{low}$ are computationally indistinguishable for QPT adversaries.
    \end{enumerate}
\end{definition}

For the two-mode specialization, choose a robust sparse-recovery matrix
$\bC\in\F_2^{\ell\times m}$ and publish
\[
    \mathbf{B}^{\mathrm{low}}=\bS\bA\oplus\bE,
    \qquad
    \mathbf{B}^{\mathrm{high}}
        =\bS\bA\oplus\bE\oplus\bC.
\]
Evaluation is the linear map
\[
    F^b(\bx)=(\bA\bx,\mathbf{B}^b\bx).
\]

The matrix $\bC$ may be obtained as follows.  Choose a parity-check matrix
$\mathbf{H}$ with $\Theta(s\log(m/s))$ rows for a binary code of distance greater
than $2s$ whose syndrome decoder recovers errors of weight at most $s$ in
polynomial time.  A shortened binary BCH code gives such a choice in our
parameter range.  Encode the syndrome $\mathbf{H}\bx$ using an asymptotically good
outer code of constant rate with a polynomial-time decoder, and let $\bC$
be the resulting linear map.  The outer decoder first recovers $\mathbf{H}\bx$
from $\bC\bx\oplus\be$ whenever $|\be|\le\gamma\ell$, for a constant
$\gamma>0$, and the syndrome decoder then recovers every $\bx$ with
$|\bx|\le s$.  This uses $\ell=\Theta(s\log(m/s))$ rows and extends the
regular-sparse presentation of Dao--Jain to the full Hamming ball needed
here.

\begin{lemma}[Dense-Sparse \textsf{LPN} lossy pair on a Hamming ball]
    \label{lem:densesparsepair}
    Assume Conjecture~\ref{conj:densesparselpn}. Then for every
    constant $\Gamma > 1$ there exist polynomials $m(\lambda)$,
    $r_A(\lambda)$, $r_B(\lambda)$, and $s(\lambda)$ with
    $s(\lambda)=\omega(\log \lambda)$ and
    $s(\lambda)=o(m(\lambda))$, together with efficiently samplable jointly
    distributed matrices
    \[
        (\bA_\lambda,
        \mathbf{B}_\lambda^{\mathrm{low}},
        \mathbf{B}_\lambda^{\mathrm{high}}),
    \]
    where $\bA_\lambda \in \F_2^{r_A(\lambda) \times m(\lambda)}$ and
    $\mathbf{B}_\lambda^{\mathrm{low}},
    \mathbf{B}_\lambda^{\mathrm{high}}
    \in \F_2^{r_B(\lambda) \times m(\lambda)}$, such that the maps
    \[
        F_\lambda^{\mathrm{low}}(\bx)
        =
        (\bA_\lambda\bx,\mathbf{B}_\lambda^{\mathrm{low}}\bx),
        \qquad
        F_\lambda^{\mathrm{high}}(\bx)
        =
        (\bA_\lambda\bx,\mathbf{B}_\lambda^{\mathrm{high}}\bx)
    \]
    satisfy the following for the set
    \[
        T_\lambda
        =
        \{\bx\in\bits^{m(\lambda)}: |\bx|\le s(\lambda)\},
        \qquad
        \ell_\lambda=\log_2 |T_\lambda|.
    \]
    \begin{enumerate}
        \item $F_\lambda^{\mathrm{high}}$ is injective on $T_\lambda$ with overwhelming probability.
        \item $|F_\lambda^{\mathrm{low}}(T_\lambda)|\le 2^{\ell_\lambda/\Gamma}$ with overwhelming probability.
        \item The public descriptions $(\bA_\lambda,\mathbf{B}_\lambda^{\mathrm{low}})$ and $(\bA_\lambda,\mathbf{B}_\lambda^{\mathrm{high}})$ are computationally indistinguishable for QPT adversaries.
    \end{enumerate}
\end{lemma}

\begin{proof}
    Fix a sufficiently large constant $D>\Gamma$ and apply
    \Cref{conj:densesparselpn} at the even dimension
    $d=2\lceil\lambda/2\rceil$.
    Write $\bA=\mathbf{T}\mathbf{M}$ as in
    \Cref{def:densesparselpn}.  If $|\bx|\le s$, then
    $|\mathbf{M}\bx|\le ks$.  Consequently the first output block takes at
    most
    \[
        |\mathcal{Y}_A|
        \le \sum_{j\le ks}\binom dj
    \]
    values.  For all sufficiently large $d$, we have
    $ks\le d/2$ and $s\le m/2$, and hence the standard binomial estimates
    give
    \[
        |T_\lambda|
        \ge \binom ms
        \ge \left(\frac ms\right)^s,
        \qquad
        |\mathcal{Y}_A|
        \le 2^{ks\log_2(ed/(ks))+1}.
    \]
    The parameter relation
    $m s^{Dk-1}\ge d^{Dk}$ therefore implies
    \[
        \frac{|T_\lambda|}{|\mathcal{Y}_A|^D}
        \ge
        2^{-D}
        \left(
            \left(\frac{k}{e}\right)^{Dk}
            \frac{m s^{Dk-1}}{d^{Dk}}
        \right)^s
        \ge
        2^{-D}\left(\frac{k}{e}\right)^{Dks}
        >1.
    \]
    Thus the first output block takes at most
    $|T_\lambda|^{1/D}$ values.

    For a fixed $\bx$ of weight at most $s$, every coordinate of $\bE\bx$
    is Bernoulli with parameter at most $\varepsilon s$.  Choose the constant
    in $\varepsilon=c/s$ sufficiently small and the constant in
    $\ell=C s\log(m/s)$ sufficiently large.  A Chernoff bound then gives
    \[
        \Pr[|\bE\bx|>\gamma\ell]
        \le 2^{-(2+o(1))s\log_2(m/s)}.
    \]
    Since $|T_\lambda|\le
    2^{(1+o(1))s\log_2(m/s)}$, a union bound shows that, except with
    negligible probability, $|\bE\bx|\le\gamma\ell$ simultaneously for
    every $\bx\in T_\lambda$.  On this event, after $\bA\bx$ is fixed,
    $\mathbf{B}^{\mathrm{low}}\bx=\bS\bA\bx\oplus\bE\bx$ lies in a Hamming
    ball of radius $\gamma\ell$.  Choosing $\gamma$ sufficiently small and
    $D$ sufficiently large therefore gives
    \[
        |F^{\mathrm{low}}(T_\lambda)|
        \le
        \left(\sum_{j\le ks}\binom dj\right)
        \left(\sum_{j\le\gamma\ell}\binom\ell j\right)
        \le |T_\lambda|^{1/\Gamma}.
    \]

    In the high branch, the trapdoor $\bS$ transforms an output into
    \[
        \mathbf{B}^{\mathrm{high}}\bx
        \oplus\bS(\bA\bx)
        =\bC\bx\oplus\bE\bx.
    \]
    The robust syndrome decoder described above recovers every
    $\bx\in T_\lambda$ on the same uniform noise event, proving injectivity.
    Finally, each branch is indistinguishable from $(\bA,\mathbf{U})$ by the
    polynomial-parallel Dense-Sparse \textsf{LPN} assumption (adding the
    fixed matrix $\bC$ preserves uniformity), and hence the two branches are
    mutually indistinguishable.
\end{proof}

\section{A bounded-fan-in, bounded-fan-out randomized encoding for dense linear maps}
\label{sec:retheorem}

Randomized encodings are a standard tool for reducing locality in low-complexity cryptography; see, for example,~\cite{Applebaum14,EPRINT:Applebaum17,JC:AppIshKus18,ICALP:AIKP15}. In this section, we give a perfect randomized encoding in the sense of \Cref{def:randenc}, with constant fan-in and fan-out in the sense of \Cref{def:faninout}. We will then use this encoding to obtain pseudoentanglement in constant depth. In addition, the same encoding has a direct cryptographic corollary: it yields a bounded-fan-in, bounded-fan-out collision-resistant hash under a random-code bounded-syndrome-decoding assumption; see \appref{app:crh}. \\

\noindent\textbf{Theorem~\ref{thm:reintro} (restated).}
\begin{itshape}
For every matrix $\mathbf{N}\in\F_2^{q\times m}$, there is an explicit perfect randomized encoding
\[
    \mathsf{RE}_{\mathbf{N}}(\bx;\br,\bs)=(\bw,\bz)
\]
of the map $\bx\mapsto \mathbf{N}\bx$ with the following properties:
\begin{enumerate}
    \item there is a deterministic decoder that recovers $\mathbf{N}\bx$ from $\mathsf{RE}_{\mathbf{N}}(\bx;\br,\bs)$, and there is an efficient simulator whose output distribution on input $\mathbf{N}\bx$ is exactly the distribution of $\mathsf{RE}_{\mathbf{N}}(\bx;R,S)$ for uniform masks $R$ and $S$;
    \item every output bit depends on at most three source bits, and every source bit influences at most three output bits;
    \item for every random variable $X$ over $\bits^m$, if $R\leftarrow\bits^{mq}$ and $S\leftarrow\bits^{q(m-1)}$ are uniform and independent of $X$, then
    \[
        H(\mathsf{RE}_{\mathbf{N}}(X;R,S))
        =
        mq+q(m-1)+H(\mathbf{N}X).
    \]
\end{enumerate}
\end{itshape}

\begin{proof}
    Formally, for an input $\bx\in\bits^m$, let
    \[
        \br=(\br_{j,i})_{j\in[m],\,i\in[q]}
        \in\bits^{mq}
    \]
    and
    \[
        \bs=(\bs_{i,j})_{i\in[q],\,j\in[m-1]}
        \in\bits^{q(m-1)}
    \]
    be uniformly random masks. We use the boundary convention
    \[
        \bs_{i,0}=\bs_{i,m}=0
        \qquad\text{for every } i\in[q].
    \]
    Define the encoding $\mathsf{RE}_{\mathbf{N}}(\bx;\br,\bs)=(\bw,\bz)$ by
    \[
        \bw_{j,0}=\bx_j\oplus \br_{j,1},
        \qquad
        \bw_{j,i}=\br_{j,i}\oplus \br_{j,i+1}
        \quad (i\in[q-1]),
    \]
    and
    \[
        \bz_{i,j}=\bs_{i,j-1}\oplus \bs_{i,j}\oplus \mathbf{N}_{i,j}\br_{j,i}
        \qquad (i\in[q],\;j\in[m]).
    \]
    Let us provide some intuition.
    The encoding uses two kinds of hiding masks, with different roles. The $\br$-masks are row-specific one-time pads for the input bits.  For each input coordinate $j$, the bits $\br_{j,1},\ldots,\br_{j,q}$ let row $i$ work with the masked value $\bx_j\oplus \br_{j,i}$.  The $\bw$-register gives a local ``difference encoding'' of these masked values: from $\bw$ one can recover $\bx_j\oplus \br_{j,i}$ for every row $i$, but not $\bx_j$ or $\br_{j,i}$ separately. This will be important when showing the privacy property.

    The $\bs$-masks have a different purpose. They hide the individual row contributions $\mathbf{N}_{i,j}\br_{j,i}$ while preserving their row parity. This is also essential for privacy. Indeed, if we instead were to set $\bz_{i,j}=\mathbf{N}_{i,j}\br_{j,i}$, then by first recovering $\bx_j\oplus \br_{j,i}$ from $\bw$ and then noting that $\bz_{i,j}=\br_{j,i}$ whenever $\mathbf{N}_{i,j}=1$, one can XOR the two in order to learn $\bx_j$.

    A fact about the $\bs$-masks that will be used in the proof is the following. For a fixed row $i$, define
    \[
        \Delta_i(\bs_{i,1},\ldots,\bs_{i,m-1})
        =
        (\bs_{i,0}\oplus \bs_{i,1},\;\bs_{i,1}\oplus \bs_{i,2},\ldots,\;\bs_{i,m-1}\oplus \bs_{i,m}),
    \]
    with $\bs_{i,0}=\bs_{i,m}=0$.  The image of $\Delta_i$ is exactly the even-parity subspace
    \[
        E_m=\left\{\bu\in\bits^m:\bigoplus_{j=1}^m \bu_j=0\right\},
    \]
    and $\Delta_i$ is a bijection from $\bits^{m-1}$ to $E_m$.  To see this, note that the parity of $\Delta_i(\bs_i)$ is zero because all internal $\bs$-variables cancel. Conversely, if $\bu\in E_m$, then the unique preimage is given by
    \[
        \bs_{i,j}=\bu_1\oplus \bu_2\oplus\cdots\oplus \bu_j,
        \qquad j\in[m-1].
    \]
    Therefore, for any fixed vector $\mathbf{b}\in\bits^m$, the random vector
    \[
        \mathbf{b}\oplus\Delta_i(S_{i,1},\ldots,S_{i,m-1})
    \]
    is uniform over the affine parity class
    \[
        \left\{\bz_i\in\bits^m:
        \bigoplus_{j=1}^m \bz_{i,j}=\bigoplus_{j=1}^m \mathbf{b}_j
        \right\}.
    \]
    This is the precise sense in which the $\bs$-masks hide everything about $\bz_i$ except its parity.

    We now show all the relevant properties of this randomized encoding.

    \smallskip
    \noindent\emph{Locality.}
    Every output bit depends on at most three source bits: $\bw_{j,0}$ depends on $\bx_j$ and $\br_{j,1}$, $\bw_{j,i}$ depends on $\br_{j,i}$ and $\br_{j,i+1}$, and $\bz_{i,j}$ depends on $\bs_{i,j-1}$, $\bs_{i,j}$, and possibly $\br_{j,i}$ according to the public coefficient $\mathbf{N}_{i,j}$. Conversely, every $\bx_j$ influences only $\bw_{j,0}$; every $\br_{j,i}$ influences at most two $\bw$-bits and one $\bz$-bit; and every $\bs_{i,j}$ influences exactly two $\bz$-bits. This proves the bounded fan-in and bounded fan-out claim.

    \smallskip
    \noindent\emph{Decodability.}
    For each $j\in[m]$ and $i\in[q]$, define
    \[
        a_{j,i}(\bw)
        =
        \bw_{j,0}\oplus \bw_{j,1}\oplus\cdots\oplus \bw_{j,i-1}.
    \]
    In a real encoding,
    \[
        a_{j,i}(\bw)=\bx_j\oplus \br_{j,i},
    \]
    because the intermediate mask bits telescope:
    \[
        (\bx_j\oplus \br_{j,1})
        \oplus(\br_{j,1}\oplus \br_{j,2})
        \oplus\cdots\oplus(\br_{j,i-1}\oplus \br_{j,i})
        =\bx_j\oplus \br_{j,i}.
    \]
    Similarly, for every row $i$,
    \[
        \bigoplus_{j=1}^m \bz_{i,j}
        =
        \bigoplus_{j=1}^m \mathbf{N}_{i,j}\br_{j,i},
    \]
    because the path masks $\bs_{i,1},\ldots,\bs_{i,m-1}$ cancel in the row parity. Therefore the decoder computes, for each row $i$,
    \[
    \begin{aligned}
        \bigoplus_{j=1}^m \bz_{i,j}
        \oplus
        \bigoplus_{j=1}^m \mathbf{N}_{i,j}a_{j,i}(\bw)
        & = \bigoplus_{j=1}^m \mathbf{N}_{i,j}\br_{j,i}
        \oplus
        \bigoplus_{j=1}^m \mathbf{N}_{i,j}(\bx_j\oplus \br_{j,i}) \\
        & = \bigoplus_{j=1}^m \mathbf{N}_{i,j}\bx_j
        =
        (\mathbf{N}\bx)_i.
    \end{aligned}
    \]
    Applying this formula to all rows deterministically recovers $\mathbf{N}\bx$ from $(\bw,\bz)$.

    \smallskip
    \noindent\emph{Perfect privacy.}
    Let $\by=\mathbf{N}\bx$.  The simulator is given $\by$ and the public matrix $\mathbf{N}$.  It proceeds as follows.
    \begin{enumerate}
        \item Sample $\bw\leftarrow\bits^{mq}$ uniformly.
        \item Compute $a_{j,i}(\bw) = \bw_{j,0}\oplus \bw_{j,1}\oplus\cdots\oplus \bw_{j,i-1}$ for every $j\in[m]$ and $i\in[q]$.
        \item For each row $i$, set the target parity
        \[
            \bp_i
            =
            \by_i\oplus\bigoplus_{j=1}^m \mathbf{N}_{i,j}a_{j,i}(\bw).
        \]
        Then sample $\bz_{i,1},\ldots,\bz_{i,m-1}$ uniformly and independently, and set
        \[
            \bz_{i,m}=\bp_i\oplus \bz_{i,1}\oplus\cdots\oplus \bz_{i,m-1}.
        \]
        When $m=1$, this means simply setting $\bz_{i,1}=\bp_i$.  This samples $\bz_i$ uniformly over all $m$-bit strings of parity $\bp_i$.
    \end{enumerate}
    The simulator outputs $(\bw,\bz)$.

    We now prove that this distribution is exactly the real one.  First, for every fixed $\bx$, the map $\br\mapsto\bw,$ in the randomized encoding, is a bijection.  Indeed, once $\bx$ and $\bw$ are fixed,
    \[
        \br_{j,i}=\bx_j\oplus a_{j,i}(\bw)
        \qquad (j\in[m],\;i\in[q]).
    \]
    Thus the real $W$ is uniform on $\bits^{mq}$ and its distribution is independent of $\bx$, matching the simulator's first step.

    Next fix $\bx$ and a value $\bw$ of $W$.  Then all $\br_{j,i}$ are fixed by the above formula.  For row $i$, define the fixed vector
    \[
        \mathbf{b}_i=(\mathbf{N}_{i,1}\br_{1,i},\ldots,\mathbf{N}_{i,m}\br_{m,i})\in\bits^m.
    \]
    The real row $Z_i$ is
    \[
        Z_i=\mathbf{b}_i\oplus\Delta_i(S_{i,1},\ldots,S_{i,m-1}).
    \]
    By the elementary fact above, this is uniform over all strings with parity
    \[
    \begin{aligned}
        \bigoplus_{j=1}^m \mathbf{N}_{i,j}\br_{j,i}
        & =
        \bigoplus_{j=1}^m \mathbf{N}_{i,j}(\bx_j\oplus a_{j,i}(\bw)) \\
        & =
        \by_i\oplus\bigoplus_{j=1}^m \mathbf{N}_{i,j}a_{j,i}(\bw)
        =\bp_i.
    \end{aligned}
    \]
    The rows are independent because the $S_{i,1},\ldots,S_{i,m-1}$ masks are independent from row to row.  Hence, conditioned on $W=\bw$, the real distribution of $Z$ is exactly the product distribution sampled by the simulator, and it depends on $\bx$ only through $\by=\mathbf{N}\bx$.  This proves perfect privacy.

    \smallskip
    \noindent\emph{Entropy identity.}
    Let $X$ be an arbitrary random variable over $\bits^m$, and let $R$ and $S$ be uniform and independent of $X$.  Write
    \[
        (W,Z)=\mathsf{RE}_{\mathbf{N}}(X;R,S).
    \]
    The same bijection $\br\mapsto\bw$ used above shows that, even when $X$ is not uniform, $W$ is uniform on $\bits^{mq}$ and independent of $X$. Therefore
    \[
        H(W)=mq.
    \]

    Fix a value $W=\bw$.  Since $W$ is independent of $X$, conditioning on $W=\bw$ does not change the distribution of $X$.  Define the row-parity random variables
    \[
        P_i=\bigoplus_{j=1}^m Z_{i,j},
        \qquad i\in[q],
    \]
    and define
    \[
        \boldsymbol{\sigma}_i(\bw)=\bigoplus_{j=1}^m \mathbf{N}_{i,j}a_{j,i}(\bw).
    \]
    From the calculation in the privacy proof,
    \[
        P_i=(\mathbf{N}X)_i\oplus\boldsymbol{\sigma}_i(\bw),
    \]
    and hence
    \[
        P=\mathbf{N}X\oplus\boldsymbol{\sigma}(\bw).
    \]
    Since $\boldsymbol{\sigma}(\bw)$ is fixed after conditioning on $W=\bw$, XORing by it is a bijection on $\bits^q$.  Therefore
    \[
        H(P\mid W=\bw)=H(\mathbf{N}X).
    \]

    Now condition further on a parity vector $P=\bp$ in the support.  Given $W=\bw$ and $P=\bp$, each row $Z_i$ is uniform over the parity class
    \[
        \left\{\bz_i\in\bits^m:
        \bigoplus_{j=1}^m \bz_{i,j}=\bp_i
        \right\},
    \]
    which has size $2^{m-1}$, and the rows are independent.  This is because, for every compatible value of $X$, the $\bs$-masks give exactly the same uniform distribution over the appropriate parity class; mixing over $X$ therefore does not change the row-internal distribution.  Thus
    \[
        H(Z\mid P=\bp, W=\bw)=q(m-1).
    \]
    Since $P$ is a deterministic function of $Z$, we have
    \[
    \begin{aligned}
        H(Z\mid W=\bw)
        &=H(P,Z\mid W=\bw) \\
        &=H(P\mid W=\bw)+H(Z\mid P,W=\bw) \\
        &=H(\mathbf{N}X)+q(m-1).
    \end{aligned}
    \]
    This holds for every $\bw$.  Averaging over $W$ and adding $H(W)=mq$ gives
    \[
        H(W,Z)=mq+q(m-1)+H(\mathbf{N}X),
    \]
    which is exactly the claimed identity.
\end{proof}

\section{Constant-depth pseudoentanglement}
\label{sec:qubitproof}

In this section, we use the randomized encoding from the previous section to construct 2D local constant-depth circuits preparing pseudoentangled states.

\vspace{0.1in}

\noindent\textbf{Theorem~\ref{thm:mainintro} (restated).}
\begin{itshape}
Assuming the intractability of Dense-Sparse \textsf{LPN} for quantum polynomial-time
(QPT) adversaries, for all sufficiently large $n > 0$ there exists an
efficiently samplable distribution over tuples
\[
    (\calC_n^{\mathrm{low}},\calC_n^{\mathrm{high}},W_n),
\]
where each $\calC_n^b$ is a 2D-local constant-depth circuit on $n$ qubits and
$W_n \subseteq [n]$, such that:
\begin{enumerate}
    \item \bfm{Additive entanglement entropy gap.}
    There exists a function $\Delta : \N \to \mathbb{R}_{\ge 1}$ such that,
    with overwhelming probability over the sampled tuple, letting
    \[
        \ket{\psi_n^{\mathrm{low}}}
        =
        \calC_n^{\mathrm{low}}\ket{0^n},
        \qquad
        \ket{\psi_n^{\mathrm{high}}}
        =
        \calC_n^{\mathrm{high}}\ket{0^n},
    \]
    and
    \[
        \rho_n^b
        =
        \Ptr{W_n}{\ket{\psi_n^b}\bra{\psi_n^b}},
        \qquad
        b \in \{\mathrm{low},\mathrm{high}\},
    \]
    one has
	    \[
	        \Se{\rho_n^{\mathrm{high}}}
	        \ge
	        \Se{\rho_n^{\mathrm{low}}} + \Delta(n).
	    \]
		    Moreover, the gap can be chosen so that $\Delta(n)=\omega(\log n)$.
	        More precisely, there are explicit thresholds
	        $a_{\mathrm{circ}}(n)<b_{\mathrm{circ}}(n)$ with
	        $b_{\mathrm{circ}}(n)-a_{\mathrm{circ}}(n)=\omega(\log n)$ such
	        that the low and high families satisfy the corresponding entropy
	        bounds except with negligible probability.

    \vspace{0.5em}

		\item \bfm{Indistinguishability.}
    For every QPT algorithm $\calA$,
    \[
        \left|
            \prob{1 \leftarrow \calA[\calC_n^{\mathrm{high}},W_n]}
            -
            \prob{1 \leftarrow \calA[\calC_n^{\mathrm{low}},W_n]}
        \right|
        =
        \negl(n),
    \]
    where the probability is over the sampled tuple and the internal
    randomness of $\calA$.
\end{enumerate}
\end{itshape}

The proof below derives this scaling from the $\lambda$-indexed
Dense-Sparse \textsf{LPN} parameters.

\begin{proof}
Fix a constant $\Gamma>1$ to be chosen later, and sample the matrices from
\Cref{lem:densesparsepair} at security parameter $\lambda$. Write
\[
    m=m(\lambda),\qquad s=s(\lambda),\qquad t=s/2,
    \qquad T=T_\lambda,
    \qquad \ell=\ell_\lambda.
\]
Let $r_A$ and $r_B$ be the row dimensions of $\bA_\lambda$ and
$\mathbf{B}_\lambda^b$, respectively, and put $q=r_A+r_B$.  Define the
row-stacked matrices
\[
    \mathbf{N}^{\mathrm{low}}
    =
    \begin{bmatrix}
        \bA_\lambda\\
        \mathbf{B}_\lambda^{\mathrm{low}}
    \end{bmatrix},
    \qquad
    \mathbf{N}^{\mathrm{high}}
    =
    \begin{bmatrix}
        \bA_\lambda\\
        \mathbf{B}_\lambda^{\mathrm{high}}
    \end{bmatrix}
    \in\F_2^{q\times m}.
\]
The injectivity and lossiness guarantees from \Cref{lem:densesparsepair} fail
only with negligible probability over the sampled matrices, so it suffices to
analyze the overwhelmingly likely event on which both guarantees hold.

Let $p=t/m=s/(2m)$ and let
\[
    X=(X_1,\ldots,X_m)\sim\Bern(p)^{\otimes m}.
\]
For $b\in\{\mathrm{low},\mathrm{high}\}$ write
\[
    Y^b=\mathbf{N}^bX.
\]
We first recall the ideal entropy gap for these two output distributions. Let
\[
    \mathcal{G}=\{|X|\le s\},
    \qquad
    \eta=\prob{\neg\mathcal{G}},
\]
and let $B$ be the indicator of $\neg\mathcal{G}$. Since
$\mathbb{E}[|X|]=t$, \Cref{prop:chernoff} gives
\[
    \eta\le \exp(-t/3)=\negl(\lambda),
\]
because $t=\omega(\log\lambda)$.

For the high branch,
\[
    H(Y^{\mathrm{high}})=H(X)-H(X\mid Y^{\mathrm{high}}).
\]
Because $B$ is determined by $X$,
\[
    H(X\mid Y^{\mathrm{high}})
    \le
    H(B)+H(X\mid Y^{\mathrm{high}},B).
\]
Conditioned on $B=0$, we have $X\in T$ and
$F_\lambda^{\mathrm{high}}$ is injective on $T$, so
$H(X\mid Y^{\mathrm{high}},B=0)=0$. Conditioned on $B=1$, the variable $X$
takes values in $\bits^m$, so
$H(X\mid Y^{\mathrm{high}},B=1)\le m$. Hence
\[
    H(Y^{\mathrm{high}})
    \ge
    H(X)-h(\eta)-\eta m.
\]
For the low branch,
\[
    H(Y^{\mathrm{low}})
    \le
    H(B)+H(Y^{\mathrm{low}}\mid B).
\]
Conditioned on $B=0$, the support of $Y^{\mathrm{low}}$ is contained in
$F_\lambda^{\mathrm{low}}(T)$, so
\[
    H(Y^{\mathrm{low}}\mid B=0)
    \le
    \log_2 |F_\lambda^{\mathrm{low}}(T)|
    \le
    \ell/\Gamma.
\]
Conditioned on $B=1$, the output is a $q$-bit string, so
$H(Y^{\mathrm{low}}\mid B=1)\le q$. Thus
\[
    H(Y^{\mathrm{low}})
    \le
    h(\eta)+(1-\eta)\ell/\Gamma+\eta q.
\]
Since $p=t/m$, $s=2t$, and $t=o(m)$,
\[
    H(X)=mh(t/m)=\Theta(t\log(m/t)),
    \qquad
    \ell=\log_2\left(\sum_{j=0}^{s}\binom{m}{j}\right)
    =\Theta(t\log(m/t)).
\]
Therefore there is an absolute constant $\alpha>0$ such that
$H(X)\ge \alpha\ell$ for all sufficiently large $\lambda$. Using also
$h(\eta)=o(\ell)$ and $\eta q=o(\ell)$, we get
\[
    H(Y^{\mathrm{high}})\ge \frac{\alpha}{2}\ell,
    \qquad
    H(Y^{\mathrm{low}})\le \frac{2}{\Gamma}\ell.
\]
Choose $\Gamma>8/\alpha$. Then
\[
    H(Y^{\mathrm{high}})-H(Y^{\mathrm{low}})
    \ge
    \left(\frac{\alpha}{2}-\frac{2}{\Gamma}\right)\ell
    \ge
    \frac{\alpha}{4}\ell.
\]
We now coherently implement the randomized encoding from
\Cref{thm:reintro}. For branch $b$, prepare source qubits
\[
    X_1,\ldots,X_m,
    \qquad
    R_{j,i}\;(j\in[m],i\in[q]),
    \qquad
    S_{i,j}\;(i\in[q],j\in[m-1]),
\]
where each $X_j$ is initialized as
\[
    \sqrt{1-p}\ket{0}+\sqrt{p}\ket{1},
\]
and each $R_{j,i}$ and $S_{i,j}$ is initialized as $\ket{+}$. Add output
qubits for the encoding bits
\[
    W_{j,0}\;(j\in[m]),
    \qquad
    W_{j,i}\;(j\in[m],i\in[q-1]),
    \qquad
    Z_{i,j}\;(i\in[q],j\in[m]),
\]
all initialized to $\ket{0}$. Using CNOT gates, compute
\[
    (W,Z)=\mathsf{RE}_{\mathbf{N}^b}(X;R,S)
\]
into these output qubits. The circuit does not uncompute the source qubits.
The resulting state is
\begin{equation}
    \label{eq:coherent-re-state}
    \ket{\psi_\lambda^b}
    =
    \sum_{\bx,\br,\bs}
    \sqrt{\prob{X=\bx}}
    \frac{1}{\sqrt{2^{mq+q(m-1)}}}
    \ket{\bx,\br,\bs}_{\mathrm{src}}
    \ket{\mathsf{RE}_{\mathbf{N}^b}(\bx;\br,\bs)}_{\mathrm{out}}.
\end{equation}
Let
\[
    W_\lambda
    =
    \{\text{all source qubits }X,R,S\}.
\]
By \Cref{prop:coherententropy}, tracing out $W_\lambda$ leaves a diagonal
state on the output register with classical distribution
$\mathsf{RE}_{\mathbf{N}^b}(X;R,S)$. Hence item~3 of \Cref{thm:reintro} gives
\begin{equation}
    \label{eq:compiled-re-entropy}
    \Se{\Ptr{W_\lambda}{\ket{\psi_\lambda^b}\bra{\psi_\lambda^b}}}
    =
    mq+q(m-1)+H(Y^b).
\end{equation}
The additive term $mq+q(m-1)$ is branch-independent. Combining
\eqref{eq:compiled-re-entropy} with the ideal entropy bounds above yields
the following explicit promise thresholds.  Set
\[
    K_\lambda=mq+q(m-1),\qquad
    a_\lambda=K_\lambda+\frac{2}{\Gamma}\ell,\qquad
    b_\lambda=K_\lambda+\frac{\alpha}{2}\ell.
\]
Then, with overwhelming probability, the low-branch reduced state has
entropy at most $a_\lambda$, the high-branch reduced state has entropy at
least $b_\lambda$, and
$b_\lambda-a_\lambda=\Omega(\ell)$.  In particular,
\[
\begin{aligned}
    &\Se{\Ptr{W_\lambda}{
        \ket{\psi_\lambda^{\mathrm{high}}}
        \bra{\psi_\lambda^{\mathrm{high}}}}}
    -
    \Se{\Ptr{W_\lambda}{
        \ket{\psi_\lambda^{\mathrm{low}}}
        \bra{\psi_\lambda^{\mathrm{low}}}}}
    \\
    &\qquad
    =
    H(Y^{\mathrm{high}})-H(Y^{\mathrm{low}})
    \ge
    \frac{\alpha}{4}\ell.
\end{aligned}
\]
Since
\[
    \ell_\lambda
    =
    \log_2\!\left(\sum_{j=0}^{2t(\lambda)}
        \binom{m(\lambda)}{j}\right)
    =
    \Theta\!\left(t(\lambda)
        \log_2\frac{m(\lambda)}{t(\lambda)}\right),
\]
the entropy gap is
\[
    \Omega\!\left(t(\lambda)
        \log_2\frac{m(\lambda)}{t(\lambda)}\right).
\]
In particular, because $t(\lambda)=\omega(\log\lambda)$ and
$t(\lambda)=o(m(\lambda))$, this gap is $\omega(\log\lambda)$.

It remains to realize the coherent encoding by nearest-neighbor gates on a
fixed 2D square-lattice template. Arrange a $q\times m$ array of constant-size
plaquettes. The plaquette at position $(i,j)$ contains the source qubits
$R_{j,i}$ and, for $j<m$, $S_{i,j}$, together with the output qubit
$Z_{i,j}$ and a constant number of routing ancillas. The $X_j$ qubit and the
output $W_{j,0}$ are placed in a constant-size cap above column $j$, while
$W_{j,i}$ for $i\in[q-1]$ is placed in a constant-size connector between the
plaquettes $(i,j)$ and $(i+1,j)$.

With this layout, the equations defining the $W$-bits are vertical nearest-
neighbor computations, and the equations defining the $Z$-bits are horizontal
nearest-neighbor computations using the path-mask convention $s_{i,0}=s_{i,m}=0$.
Each source qubit participates in at most three CNOTs, so the CNOTs can be
colored into a constant number of layers. The branch $b$ affects only whether
the CNOT from $R_{j,i}$ to $Z_{i,j}$ is included, according to the public bit
$\mathbf{N}_{i,j}^b$. The geometry, the initialization pattern, and the cut
$W_\lambda$ are the same in both branches. Thus the sampled circuits are
2D-local and constant-depth. This is illustrated in \cref{fig:2d-coherent-re} and \cref{fig:local-plaquette}.

The number of qubits is
\[
    n(\lambda)=\Theta(mq),
\]
which is polynomial in the security parameter. The only parameter-dependent
one-qubit gates are the preparations
\[
    \ket{0}\longmapsto \sqrt{1-p}\ket{0}+\sqrt{p}\ket{1},
    \qquad p=t/m,
\]
which can be implemented by
\[
    R_y(2\arcsin\sqrt{p})=R_y(\arccos(1-2p)).
\]

Finally, indistinguishability follows by a compiler argument. Let
$\mathsf{Comp}$ be the deterministic polynomial-time classical compiler that
maps a public key $(\bA,\mathbf{B})$ to the corresponding 2D-local coherent
randomized-encoding circuit together with the cut $W_\lambda$. This compiler
is identical in both branches except for the public matrix entries that decide
which CNOTs from $R_{j,i}$ to $Z_{i,j}$ are present. If a QPT algorithm
distinguished
\[
    (\calC_\lambda^{\mathrm{high}},W_\lambda)
    =
    \mathsf{Comp}(\bA_\lambda,\mathbf{B}_\lambda^{\mathrm{high}})
\]
from
\[
    (\calC_\lambda^{\mathrm{low}},W_\lambda)
    =
    \mathsf{Comp}(\bA_\lambda,\mathbf{B}_\lambda^{\mathrm{low}})
\]
with non-negligible advantage, then a QPT distinguisher for the two
Dense-Sparse \textsf{LPN} public-key distributions could compute
$\mathsf{Comp}(\bA,\mathbf{B})$ on its input public key and run that algorithm.
This contradicts \Cref{lem:densesparsepair}. Hence
\[
    \left|
        \prob{1 \leftarrow \calA[\calC_\lambda^{\mathrm{high}},W_\lambda]}
        -
        \prob{1 \leftarrow \calA[\calC_\lambda^{\mathrm{low}},W_\lambda]}
    \right|
    =
    \negl(\lambda).
\]

To obtain every sufficiently large qubit count, fix constants
$c,C>0$ such that the unpadded construction has at most $c\lambda^C$
qubits.  On input a desired size $n$, use
$\lambda=\lfloor(n/c)^{1/C}\rfloor$ (decreasing it by a constant if
necessary) and append inert qubits initialized to $\ket0$ until the circuit
has exactly $n$ qubits.  Put these qubits outside the cut and apply no gates
to them.  This preserves both entropies, the thresholds above, circuit
depth, geometric locality, and indistinguishability.  Since
$\lambda=\Omega(n^{1/C})$, negligible functions of $\lambda$ remain
negligible in $n$, and the entropy gap is still $\omega(\log n)$.
This proves \Cref{thm:mainintro} for every sufficiently large $n$.
\end{proof}

\begin{remark}
    The coherent randomized-encoding compiler intentionally adds a large
    branch-independent entropy term $mq+q(m-1)$. This is why the theorem is
    stated as an additive entropy-gap result. The construction does not
    preserve the multiplicative entropy ratio of the ideal lossy-function
    calculation. On the other hand, the additive gap is preserved exactly, and
    the circuit uses only coherent classical XOR operations after the biased
    source qubits are prepared.
\end{remark}

\section{Hardness of learning local Hamiltonian ground-state entanglement}
\label{sec:hamapp}

We prove the two Hamiltonian consequences from the introduction. The 2D construction is the direct conjugated-projector construction. The 1D construction uses the padded QTM history-state Hamiltonian from~\cite{groundstatehardness}.

\subsection{The 2D constant-gap construction}
\label{subsec:2dhamproof}

\noindent\textbf{Theorem~\ref{thm:hamintro} (restated).}
\begin{itshape}
Assuming the intractability of Dense-Sparse \textsf{LPN} for QPT adversaries, for all sufficiently large $n>0$ there exists an efficiently samplable distribution over tuples
\[
    (H_n^{low},H_n^{high},W_n),
\]
where each $H_n^b$ is a 2D $k$-local Hamiltonian on $n$ qubits for some constant $k>0$ and $W_n\subseteq[n]$, such that:
\begin{enumerate}
    \item with overwhelming probability over the sampled tuple, each Hamiltonian $H_n^b$ is frustration-free, has a unique ground state $\ket{\phi_n^b}$, and has spectral gap $1$;
    \item there exists a function $\Delta:\N\to\mathbb{R}_{\ge 1}$ with $\Delta(n)=\omega(\log n)$ such that, with overwhelming probability over the sampled tuple,
    \[
        \Se{\Ptr{W_n}{\ket{\phi_n^{high}}\bra{\phi_n^{high}}}}
        \ge
        \Se{\Ptr{W_n}{\ket{\phi_n^{low}}\bra{\phi_n^{low}}}}
        +
        \Delta(n);
    \]
    \item for every QPT algorithm $\calA$,
    \[
        \left|
            \prob{1\leftarrow \calA[H_n^{high},W_n]}
            -
            \prob{1\leftarrow \calA[H_n^{low},W_n]}
        \right|
        =
        \negl(n).
    \]
\end{enumerate}
Consequently, for the explicit thresholds $a_2,b_2$ derived below,
$\mathsf{LGHES}_{a_2,b_2}$ is quantumly hard on this family even for 2D
constant-gap local Hamiltonians.
\end{itshape}

\begin{proof}
    Let the sampled circuits $\calC_n^{low},\calC_n^{high}$, the common cut $W_n$, and the states
    \[
        \ket{\psi_n^{low}}=\calC_n^{low}\ket{0^n},
        \qquad
        \ket{\psi_n^{high}}=\calC_n^{high}\ket{0^n}
    \]
    be given by \Cref{thm:mainintro}. For $b\in\{low,high\}$ define
    \begin{equation}
        \label{eq:hcdef}
        H_n^b=
        \sum_{i=1}^n h_{n,i}^b,
        \qquad
        h_{n,i}^b=
        \calC_n^b\ket{1}\!\bra{1}_i(\calC_n^b)^\dagger.
    \end{equation}
    Equivalently,
    \[
        H_n^b=
        \calC_n^b
        \left(\sum_{i=1}^n \ket{1}\!\bra{1}_i\right)
        (\calC_n^b)^\dagger.
    \]
    Each term is positive semidefinite, and
    \[
        h_{n,i}^b\ket{\psi_n^b}
        =
        \calC_n^b\ket{1}\!\bra{1}_i\ket{0^n}=0
    \]
    for every $i$, so $H_n^b$ is frustration-free and $\ket{\psi_n^b}$ is a ground state. Unitary conjugation preserves the spectrum, and $\sum_i\ket{1}\!\bra{1}_i$ has unique ground state $\ket{0^n}$ and spectral gap $1$. Hence $H_n^b$ has unique ground state
    \[
        \ket{\phi_n^b}=\ket{\psi_n^b}
    \]
    and spectral gap $1$.

    Since each circuit $\calC_n^b$ is 2D-local and constant-depth, conjugating a one-qubit projector through the circuit expands its support only to a constant-size 2D light cone. Therefore every term $h_{n,i}^b$ is supported on at most $k$ qubits of constant lattice diameter, for a constant $k$ independent of $n$.

    The entanglement claim follows immediately from \Cref{thm:mainintro}, because the ground states are exactly the circuit output states. In particular, the lower bound $\Delta(n)=\omega(\log n)$ is inherited directly from the pseudoentangled-state entropy gap in \Cref{thm:mainintro}. The branchwise bounds in the proof of that theorem give explicit thresholds $a_2(n)<b_2(n)$, with $b_2(n)-a_2(n)=\omega(\log n)$, for which the low and high Hamiltonian families satisfy the two respective promises of $\mathsf{LGHES}_{a_2,b_2}$. The indistinguishability claim follows from a deterministic polynomial-time compiler: given $(\calC,W)$, compute the constant-size light cone of each qubit and output the local terms of $H_{\calC}$ together with $W$. If a QPT algorithm distinguished $(H_n^{high},W_n)$ from $(H_n^{low},W_n)$, composing it with this compiler would distinguish $(\calC_n^{high},W_n)$ from $(\calC_n^{low},W_n)$, contradicting \Cref{thm:mainintro}.
\end{proof}

\subsection{The 1D history-state construction}
\label{subsec:1dhistoryproof}

We next prove the 1D Hamiltonian consequence.  We use one fixed universal
quantum Turing Machine, with the circuit description placed on its classical input track.  The
machine writes the circuit output to fixed cells and then changes only its
clock and auxiliary tracks. We apply the QTM-to-Hamiltonian construction of
Bouland--Zhang--Zhou~\cite{groundstatehardness}.  Since the output cells do
not move during the padding part of the history, one physical cut inherits
the output entropy.

\begin{lemma}[1D history-state compiler]
    \label{lem:paddedhistory}
    Let $\calC$ be a polynomial-size quantum circuit on $n$ qubits whose
    gates are efficiently described, let
    $W\subseteq[n]$, and let $0<\varepsilon\le1$.  Write the circuit in a
    gate-list format in which every gate description has the same width and
    unused positions contain identity gates.
    There is a deterministic polynomial-time compiler that outputs a 1D local Hamiltonian
    $H_{\calC,\varepsilon}$ on
    $N=\poly(n,|\calC|,1/\varepsilon)$ qubits and a subset
    $W_{\calC,\varepsilon}\subseteq[N]$ such that:
    \begin{enumerate}
        \item $H_{\calC,\varepsilon}$ is frustration-free and has a unique ground state $\ket{\phi_{\calC,\varepsilon}}$;
        \item the spectral gap of $H_{\calC,\varepsilon}$ is at least $1/\poly(N)$;
        \item if $\ket{\psi}=\calC\ket{0^n}$, then
        \[
            \left|
            \Se{\Ptr{W_{\calC,\varepsilon}}{
                \ket{\phi_{\calC,\varepsilon}}
                \bra{\phi_{\calC,\varepsilon}}
            }}
            -
            \Se{\Ptr{W}{\ket{\psi}\bra{\psi}}}
            \right|
            \le
            \varepsilon.
        \]
    \end{enumerate}
    Moreover, $N$ and the physical locations used for the output depend only
    on $n$, the length of the padded gate list, and $\varepsilon$.  Thus two
    circuits with the same values of these parameters and the same set $W$
    produce Hamiltonians of the same size and the same set
    $W_{\calC,\varepsilon}$.
\end{lemma}

\begin{proof}
    Fix a finite universal gate set containing the identity, and set
    \[
        \delta=\frac{\varepsilon^2}{64n^2}.
    \]
    Compile the efficiently described gates of $\calC$ to this set so that,
    for the compiled output $\ket{\widetilde\psi}$,
    \[
        \frac12\left\|
            \ket{\psi}\!\bra{\psi}
            -\ket{\widetilde\psi}\!\bra{\widetilde\psi}
        \right\|_1
        \le \delta,
        \qquad
        \ket{\psi}=\calC\ket{0^n}.
    \]
    Pad the compiled word for every gate position to the same worst-case
    length, using identity gates.  This length depends only on $n$, the
    width and number of the original gate descriptions, and $\varepsilon$.
    The length of the resulting list is therefore fixed by these values, and
    the compilation has polynomial overhead.  Trace
    distance cannot increase under partial trace.  The Fannes--Audenaert
    bound and $h(\delta)\le 2\sqrt\delta$ give
    \[
    \begin{split}
        \left|
        \Se{\Ptr{W}{\ket{\psi}\bra{\psi}}}
        -\Se{\Ptr{W}{\ket{\widetilde\psi}
            \bra{\widetilde\psi}}}
        \right|
        &\le \delta n+h(\delta)\\
        &\le \frac{\varepsilon}{2}.
    \end{split}
    \]

    Fix a reversible multi-track universal QTM $\mathcal U$.  Its state set,
    alphabet, tracks, and transition rule do not depend on $\calC$.  Its
    classical input contains $n$ and the fixed-length compiled gate list.
    The machine starts with $n$ data cells in $\ket{0^n}$ and applies the
    listed gates in order.  Each instruction returns the head and work
    tracks to a state that is independent of the data.  After the last
    instruction, the same $n$ data cells contain
    $\ket{\widetilde\psi}$ and are never changed again.  The machine then
    enters a reversible padding loop on clock and auxiliary tracks.  Hence,
    after a number
    $T_{\mathcal U}$ of computation steps that is polynomial in $n$ and the
    padded program length, every configuration has the
    form
    \[
        \ket{\widetilde\psi}_{O}
        \ket{\mathrm{aux}_t}_{\overline O},
        \qquad t>T_{\mathcal U},
    \]
    where $O$ is the same set of output cells for every program of this
    length and $\ket{\mathrm{aux}_t}$ is independent of
    $\ket{\widetilde\psi}$.  The bound $T_{\mathcal U}$ depends only on $n$ and the
    padded program length, not on the gates in the list.

    Apply the non-translation-invariant multi-track QTM-to-Hamiltonian
    construction of Bouland--Zhang--Zhou
    \cite[Secs.~6--7]{groundstatehardness} to the fixed machine
    $\mathcal U$, using the compiled gate list as its classical input.  Fix
    the number of clock tracks large enough to dominate the running time of
    the fixed universal simulation, including the sweeps used in their
    construction.  This number is independent of the input circuit.  More
    explicitly, for a tape parameter $s$, giving a line of $s+2$ particles,
    there are fixed constants
    $c,c'>0$ and fixed integers $q<K$ such that the output has been reached
    within $T_0\le cs^q$ history configurations and the full history has
    length $L\ge c's^K$.  Increase $s$ by a polynomial amount until
    \[
        \frac{T_0}{L}\le\delta.
    \]
    The resulting $s$ is polynomial in
    $n$, the padded program length, and $1/\varepsilon$.  The padding loop
    is chosen so that the machine does not halt before the history ends.
    Their local validity, transition, and input terms give
    a nearest-neighbor Hamiltonian on a line of constant-dimensional
    particles.  The local dimension is fixed because $\mathcal U$ and the
    number of clock tracks are fixed.  For the initialized computation its
    unique frustration-free ground state is
    the uniform history state
    \[
        \ket{\Phi}
        =
        \frac1{\sqrt L}\sum_{t=1}^{L}
        \ket{\mathrm{Clock}(t)}
        \ket{\mathrm{Track}(t)},
    \]
    and its spectral gap is $1/\poly(L)$
    \cite[Thm.~6.9]{groundstatehardness}.  The line length, $L$, and all
    local terms can be computed from $n$, the padded program length, and
    $\varepsilon$ in polynomial time.

    Let $R\subseteq O$ be the fixed output cells corresponding to
    $\overline W$, and define the physical cut by taking
    $W_{\calC,\varepsilon}$ to be every physical degree of freedom except
    $R$.  Since distinct clock configurations are orthogonal, tracing out
    every system except $R$ removes all cross terms between different
    times.  Therefore, for some state $\sigma_R$ from the first $T_0$
    history steps,
    \[
        \Phi_R
        =
        \frac{T_0}{L}\sigma_R
        +\left(1-\frac{T_0}{L}\right)
        \Ptr{W}{\ket{\widetilde\psi}
            \bra{\widetilde\psi}}.
    \]
    It follows that the trace distance between these two reduced states is
    at most $T_0/L\le\delta$.  Another use of the Fannes--Audenaert bound
    gives
    \[
        \left|
            \Se{\Phi_R}
            -\Se{\Ptr{W}{\ket{\widetilde\psi}
                \bra{\widetilde\psi}}}
        \right|
        \le \delta n+h(\delta)
        \le \frac{\varepsilon}{2}.
    \]
    Together with the gate-compilation bound, this proves the claimed
    $\varepsilon$ bound.

    Finally encode each constant-dimensional chain particle into a block of
    $b$ qubits, for a fixed constant $b$, choosing the block
    isometry so that the output-track qubit is a designated physical qubit.
    This is possible because the output track is a qubit factor of the fixed
    multi-track alphabet.  Conjugate every nearest-neighbor term by these
    block isometries and add one-block penalties outside the code spaces.
    The resulting Hamiltonian has $N=b(s+2)$ qubits and is 1D $k$-local for constant $k$,
    remains frustration-free with a unique ground state, and has gap at
    least the minimum of the qudit gap and $1$.  The fixed set $R$, and
    hence $W_{\calC,\varepsilon}$, is now a literal subset of physical
    qubits.  All steps are deterministic and polynomial time.
\end{proof}

\noindent\textbf{Theorem~\ref{thm:ham1dintro} (restated).}
\begin{itshape}
Assuming the intractability of Dense-Sparse \textsf{LPN} for QPT adversaries, for all sufficiently large $n>0$ there exists an efficiently samplable distribution over tuples
\[
    (H_n^{low},H_n^{high},W_n),
\]
where each $H_n^b$ is a 1D $k$-local Hamiltonian on $n$ qubits for some constant $k>0$ and $W_n\subseteq[n]$, such that:
\begin{enumerate}
    \item with overwhelming probability over the sampled tuple, each Hamiltonian $H_n^b$ is frustration-free, has a unique ground state $\ket{\phi_n^b}$, and has spectral gap at least $1/\poly(n)$;
    \item there exists a function $\Delta:\N\to\mathbb{R}_{>0}$ with $\Delta(n)\ge 1/2$ such that, with overwhelming probability over the sampled tuple,
    \[
        \Se{\Ptr{W_n}{\ket{\phi_n^{high}}\bra{\phi_n^{high}}}}
        \ge
        \Se{\Ptr{W_n}{\ket{\phi_n^{low}}\bra{\phi_n^{low}}}}
        +
        \Delta(n);
    \]
    \item for every QPT algorithm $\calA$,
    \[
        \left|
            \prob{1\leftarrow \calA[H_n^{high},W_n]}
            -
            \prob{1\leftarrow \calA[H_n^{low},W_n]}
        \right|
        =
        \negl(n).
    \]
\end{enumerate}
For the explicit thresholds $a_1,b_1$ constructed below,
$\mathsf{LGHES}_{a_1,b_1}$ is quantumly hard on this family.
\end{itshape}

\begin{proof}
    Start from the two circuit families and cut $(\calC_m^{low},\calC_m^{high},W_m)$ of \Cref{thm:mainintro}. Let
    \[
        \ket{\psi_m^b}=\calC_m^b\ket{0^m},
        \qquad
        b\in\{low,high\}.
    \]
    For all sufficiently large $m$, the proof of
    \Cref{thm:mainintro} gives explicit thresholds
    $a_{\mathrm{circ}}(m)<b_{\mathrm{circ}}(m)$ such that
    \[
        \Se{\Ptr{W_m}{\ket{\psi_m^{low}}\bra{\psi_m^{low}}}}
        \le a_{\mathrm{circ}}(m),
        \qquad
        \Se{\Ptr{W_m}{\ket{\psi_m^{high}}\bra{\psi_m^{high}}}}
        \ge b_{\mathrm{circ}}(m),
    \]
    with
    $b_{\mathrm{circ}}(m)-a_{\mathrm{circ}}(m)=\omega(\log m)$ and in
    particular at least $1$ for all sufficiently large $m$.

    The two circuits have the same layout and the same cut.  Put one gate
    position at every location in this common layout, using the identity
    when a gate is absent.  Their padded gate lists then have the same
    length, determined only by $m$.  Since this length and the running time
    are bounded by fixed polynomials in $m$, choose a fixed integer $r$ so
    that the QTM tape parameter
    \[
        s(m)=m^r
    \]
    contains the program and gives
    \[
        \frac{T_0}{L}\le \frac{1}{6400m^2},
    \]
    which is the value of $\delta$ in the lemma when
    $\varepsilon=1/10$.  Increase $r$ once if needed; it is then the same
    for every $m$ and both branches.  Use the construction in
    \Cref{lem:paddedhistory} with this tape length.  If the block encoding
    in its proof uses $b$ qubits
    per QTM particle, the two outputs have the same number
    \[
        N_0(m)=b(m^r+2)
    \]
    of qubits and the same physical cut $W_{N_0(m)}$.  Let their unique ground states be
    $\ket{\phi_{N_0(m)}^{low}}$ and
    $\ket{\phi_{N_0(m)}^{high}}$.  Each Hamiltonian is frustration-free and
    has inverse-polynomial gap, and
    \[
        \left|
        \Se{\Ptr{W_{N_0(m)}}{\ket{\phi_{N_0(m)}^b}
            \bra{\phi_{N_0(m)}^b}}}
        -
        \Se{\Ptr{W_m}{\ket{\psi_m^b}\bra{\psi_m^b}}}
        \right|
        \le \frac{1}{10}
    \]
    for $b\in\{low,high\}$.  Thus, with overwhelming probability,
    \[
    \begin{aligned}
        \Se{\Ptr{W_{N_0(m)}}{\ket{\phi_{N_0(m)}^{low}}
            \bra{\phi_{N_0(m)}^{low}}}}
        &\le a_{\mathrm{circ}}(m)+\frac{1}{10},\\
        \Se{\Ptr{W_{N_0(m)}}{\ket{\phi_{N_0(m)}^{high}}
            \bra{\phi_{N_0(m)}^{high}}}}
        &\ge b_{\mathrm{circ}}(m)-\frac{1}{10}.
    \end{aligned}
    \]
    In particular, their entropy difference is at least
    \[
        b_{\mathrm{circ}}(m)-a_{\mathrm{circ}}(m)-\frac15
        \ge \frac12
    \]
    for all sufficiently large $m$.

    The map
    \[
        (\calC,W)\longmapsto
        (H_{\calC,1/10},W_{\calC,1/10})
    \]
    first puts $\calC$ in the fixed-length format above and is deterministic
    and polynomial time.  A QPT distinguisher for the two Hamiltonian
    descriptions would therefore distinguish the two circuit descriptions
    from \Cref{thm:mainintro}.

    It remains to obtain every system size.  For each sufficiently large
    requested size $n$, define
    \[
        m(n)=\max\{m:N_0(m)\le n\},
    \]
    where the maximum is over the sufficiently large parameters used above.
    Since $N_0(m)=b(m^r+2)$ is increasing, this value can be computed
    directly and satisfies $m(n)=\Theta(n^{1/r})$.

    Sample and compile the pair for $m=m(n)$, and append
    $n-N_0(m)$ qubits to the right end of each chain.  Add the one-local
    projector $\ket1\!\bra1$ on each new qubit and include the new qubits in
    $W_n$.  The resulting Hamiltonians have exactly $n$ qubits and the
    unique ground states
    \[
        \ket{\phi_{N_0(m)}^b}\ket{0}^{\otimes(n-N_0(m))}.
    \]
    This does not change either entropy.  It preserves frustration freeness
    and changes the gap only to the minimum of the old gap and $1$, which is
    at least $1/\poly(n)$.  Define
    \[
        a_1(n)=a_{\mathrm{circ}}(m(n))+\frac{1}{10},
        \qquad
        b_1(n)=b_{\mathrm{circ}}(m(n))-\frac{1}{10}.
    \]
    The bounds above give the two promises.  Taking
    $\Delta(n)=b_1(n)-a_1(n)$ gives $\Delta(n)\ge1/2$.

    Finally, the map from a circuit at parameter $m(n)$ to the padded
    $n$-qubit Hamiltonian is deterministic and polynomial time.  Since
    $m(n)\ge n^{\Omega(1)}$, a non-negligible distinguishing advantage in
    $n$ would be non-negligible in $m(n)$.  This proves indistinguishability,
    the stated $\mathsf{LGHES}_{a_1,b_1}$ hardness, and the theorem for every
    sufficiently large $n$.
\end{proof}

\bibliographystyle{alpha}
\bibliography{add}

\appendix
\makeatletter
\renewcommand*{\theHsection}{appendix.\Alph{section}}
\renewcommand*{\theHsubsection}{appendix.\Alph{section}.\arabic{subsection}}
\makeatother

\addtocontents{toc}{\protect\setcounter{tocdepth}{1}}

\section{Poor man's GHZ states and an alternative pseudoentanglement construction}
\label{app:pmghz}

The main text uses a randomized encoding for the pseudoentanglement construction. We include here an alternate construction based on poor man's GHZ states. Poor man's GHZ states are GHZ states with an unknown Pauli $X$ offset recorded in an auxiliary shift register. In contrast to ordinary GHZ states, PMGHZs can be prepared in constant depth, as shown in~\cite{watts}. Here, we will consider biased versions of these states, meaning that the two components of the GHZ superposition are not equally weighted.

\begin{definition}[Poor man's GHZ states]
    \label{def:pmghz}
    For $\ell>0$, $\eps\in[0,1]$, and $\bz\in\bits^\ell$, define the biased poor man's GHZ state
    \[
        \ket{PMGHZ_{\eps,\ell}(\bz)}
        =
        \sqrt{1-\eps}\ket{\bz}
        +
        \sqrt{\eps}\ket{\bz\oplus 1^\ell}.
    \]
    When $\eps=1/2$, this is the unbiased poor man's GHZ state.
\end{definition}

We now restate the proof from~\cite{watts} that such states can be prepared in constant depth on a 1D architecture, adapting it to the biased version of these states.

\begin{proposition}[Poor man's GHZ preparation~\cite{watts}]
    \label{prop:pmghz}
    For every $\ell>0$ and every $\eps\in[0,1]$, there exists a 1D-local constant-depth quantum circuit, using only one- and two-qubit gates, whose output on $2\ell-1$ qubits can be written as
    \[
        \frac{1}{\sqrt{2^{\ell-1}}}
        \sum_{\bw\in\bits^{\ell-1}}
        \ket{PMGHZ_{\eps,\ell}(\bz(\bw))}\ket{\bw},
    \]
    where $\bz(\bw)\in\bits^\ell$ is determined by $\bz_1(\bw)=0$ and
    \[
        \bz_{i+1}(\bw)=\bz_i(\bw)\oplus \bw_i
        \qquad\text{for all } i\in[\ell-1].
    \]
    In other words, the adjacent differences of $\bz(\bw)$ are the bits of $\bw$.
\end{proposition}

\begin{proof}
    Arrange $2\ell-1$ qubits on a line as
    \[
        \br_1,\mathbf{b}_1,\br_2,\mathbf{b}_2,\ldots,\mathbf{b}_{\ell-1},\br_\ell,
    \]
    where the $\br_i$ will form the logical register (i.e.\ the qubits that will comprise the PMGHZ state) and the $\mathbf{b}_i$ the shift register (i.e.\ the qubits encoding the parities). Initialize the qubits in the state
    \[
        \left(\sqrt{1-\eps}\ket{0}_{\br_1}+\sqrt{\eps}\ket{1}_{\br_1}\right)
        \otimes
        \ket{+}_{\br_2}\otimes\cdots\otimes\ket{+}_{\br_\ell}
        \otimes
        \ket{0}_{\mathbf{b}_1}\otimes\cdots\otimes\ket{0}_{\mathbf{b}_{\ell-1}}.
    \]
    The first logical qubit can be prepared from $\ket{0}$ by
    \[
        R_y(2\arcsin\sqrt{\eps})=R_y(\arccos(1-2\eps)).
    \]
    Apply two layers of nearest-neighbor CNOT gates:
    \begin{enumerate}
        \item for each $i\in[\ell-1]$, apply $\mathsf{CNOT}(\br_i,\mathbf{b}_i)$;
        \item for each $i\in[\ell-1]$, apply $\mathsf{CNOT}(\br_{i+1},\mathbf{b}_i)$.
    \end{enumerate}
    Within each layer the gates act on disjoint neighboring pairs. Let
    \[
        \alpha_0=\sqrt{1-\eps},\qquad \alpha_1=\sqrt{\eps},
    \]
    and for $\bu=(\bu_1,\ldots,\bu_\ell)$ define
    \[
        \Delta(\bu)=(\bu_1\oplus \bu_2,\ldots,\bu_{\ell-1}\oplus \bu_\ell).
    \]
    Expanding the $\ket{+}$ states, the final state is
    \[
        \frac{1}{\sqrt{2^{\ell-1}}}
        \sum_{\bu\in\bits^\ell}
        \alpha_{\bu_1}\ket{\bu}\ket{\Delta(\bu)}.
    \]
    For every shift string $\bw$, the two strings satisfying $\Delta(\bu)=\bw$ are $\bz(\bw)$ and $\bz(\bw)\oplus 1^\ell$, where $\bz_1(\bw)=0$. Regrouping the sum by $\bw$ gives the claimed expression.
    Note that if one were to measure the shift register, the logical register would collapse to one PMGHZ state whose shift will be determined by the measurement outcome.
\end{proof}

With this, we are ready to give the alternate constant-depth pseudoentanglement construction.

\subsection{Mirrored-shift PMGHZ construction}
\label{app:pmghzcompiler}

Recall that the main construction, from \Cref{sec:qubitproof}, used a constant-depth randomized encoding for linear mappings of the form $\bx\mapsto \mathbf{N}\bx$. Here, we will effectively do the same thing, except the encoding will be fundamentally quantum and will rely on PMGHZ states. The intuition is that the linear mapping could be implemented in constant depth if we had sufficiently many copies of each input bit and sufficiently many copies of each output bit, since all pairwise interactions could then be parallelized. GHZ states effectively provide many coherent copies of one label bit, meaning the bit that distinguishes the two components $\ket{0^\ell}$ and $\ket{1^\ell}$ of a GHZ block. They also allow parity computations in phase: if one acts with $Z^a$ and $Z^b$ on distinct qubits of a GHZ state, the relative phase becomes $(-1)^{a\oplus b}$, effectively performing a nonlocal parity computation. Since the input bits in our application have bias $t/m$, we would have to use biased GHZ states for the input blocks. Ordinary GHZ states are not preparable in constant depth, so we substitute PMGHZ states. The point of the mirrored-shift construction below is to make the resulting Pauli $X$ offsets removable by unitaries that are local with respect to the entanglement cut, as this will leave the entropy unchanged.

More formally, for each input coordinate $j\in[m]$, prepare a biased PMGHZ block on logical sites
\[
    \br_{j,0},\br_{j,1},\ldots,\br_{j,q},
\]
with bias $p=t/m$. The distinguished site $\br_{j,0}$ stores the logical input bit, while $\br_{j,i}$ is used in row $i$. For each output row $i\in[q]$, prepare an unbiased PMGHZ block on sites
\[
    \bo_{i,1},\ldots,\bo_{i,m}.
\]
Let $\bu$ denote all input-block shift strings and $\bv$ all output-block shift strings. The cut used for this construction is the following fixed bipartition. Let $A$ be the input side and $B$ be the output side. Before mirroring, $A$ contains all input logical sites $\br_{j,i}$ and the original input-block shift registers $\bu$, while $B$ contains all output logical sites $\bo_{i,j}$ and the original output-block shift registers $\bv$. To make the offset corrections local across this cut, add mirror registers: copy every input shift bit coherently to a fresh qubit $\bu^{\mathrm{mir}}$ placed on $B$, and copy every output shift bit coherently to a fresh qubit $\bv^{\mathrm{mir}}$ placed on $A$. These are CNOTs from computational-basis shift registers; they are not measurements. Thus the final cut is
\[
    A
    =
    \{\br_{j,i}\}_{j,i}
    \cup
    \{\text{original input shifts }\bu\}
    \cup
    \{\text{mirror copies }\bv^{\mathrm{mir}}\},
\]
and
\[
    B
    =
    \{\bo_{i,j}\}_{i,j}
    \cup
    \{\text{original output shifts }\bv\}
    \cup
    \{\text{mirror copies }\bu^{\mathrm{mir}}\}.
\]
This is the cut across which all entropies are computed. It is fixed and the same for the low and high branches. If $K_{\mathrm{PM}}$ is the total number of shift bits copied across the cut, then, after local offset-removal unitaries within $A$ and within $B$, the mirror registers contribute exactly $K_{\mathrm{PM}}$ ebits across the cut.

For every matrix entry $\mathbf{N}_{i,j}=1$, apply a $\mathsf{CZ}$ gate between $\br_{j,i}$ and $\bo_{i,j}$. Conditioned on shifts, write the input offset at this position as $a_{j,i}(\bu_j)$ and the output offset as $b_{i,j}(\bv_i)$. If $\bx_j$ is the label bit of the input PMGHZ block and $\by_i$ is the label bit of output PMGHZ block $i$ (that is, the bit selecting which of the two coherent components of the corresponding PMGHZ block is present before the shift offset is applied), then the corresponding phase exponent is
\[
    \mathbf{N}_{i,j}(\bx_j\oplus a_{j,i}(\bu_j))(\by_i\oplus b_{i,j}(\bv_i)).
\]
Summing over all entries gives
\[
\begin{aligned}
    &\sum_{i,j}\mathbf{N}_{i,j}(\bx_j\oplus a_{j,i}(\bu_j))(\by_i\oplus b_{i,j}(\bv_i))\\
    &\qquad=
    \by\cdot \mathbf{N}\bx
    \oplus
    \by\cdot\boldsymbol{\sigma}(\bu)
    \oplus
    \bx\cdot\boldsymbol{\tau}(\bv)
    \oplus
    \kappa(\bu,\bv),
\end{aligned}
\]
where
\[
    \boldsymbol{\sigma}_i(\bu)=\bigoplus_j \mathbf{N}_{i,j}a_{j,i}(\bu_j),
    \qquad
    \boldsymbol{\tau}_j(\bv)=\bigoplus_i \mathbf{N}_{i,j}b_{i,j}(\bv_i),
\]
and
\[
    \kappa(\bu,\bv)=\bigoplus_{i,j}\mathbf{N}_{i,j}a_{j,i}(\bu_j)b_{i,j}(\bv_i).
\]
The useful part of the first term is the bit string $\mathbf{N}\bx$: the coefficient of the auxiliary output PMGHZ label bit $\by_i$ is
\[
    \bigoplus_j \mathbf{N}_{i,j}\bx_j
    =
    (\mathbf{N}\bx)_i.
\]
Thus, before accounting for PMGHZ offsets, the retained output label qubits are a phase encoding of $\mathbf{N}\bx$. The other terms are offset-induced corrections. Because of the mirror registers, side $A$ has local access to $\bv$ through $\bv^{\mathrm{mir}}$ and can remove $\bx\cdot\boldsymbol{\tau}(\bv)$ by controlled-$Z$ phase corrections, while side $B$ has local access to $\bu$ through $\bu^{\mathrm{mir}}$ and can remove $\by\cdot\boldsymbol{\sigma}(\bu)\oplus\kappa(\bu,\bv)$ by controlled-$Z$ phase corrections using its original $\bv$ registers as well. Controlled-$X$ gates, controlled by the shift registers on the same side of the cut, then remove the $X$-offsets inside the PMGHZ blocks. Finally, within each input block, CNOTs from the distinguished qubit $\br_{j,0}$ to the other logical sites map the repetition state $\ket{x_j}^{\otimes(q+1)}$ to $\ket{x_j}\ket{0^q}$; within each output block, CNOTs from one chosen representative qubit to the other output logical sites similarly keep a single label qubit and reset the rest. Applying Hadamards to these retained output label qubits converts the Fourier-basis encoding into the computational-basis value $\ket{\mathbf{N}\bx}$. These operations are used only to analyze the entanglement entropy: they show local equivalence across the cut. Thus the compiled state is locally equivalent to
\[
    \ket{\Phi_{\mathbf{N}}}
    \otimes
    \ket{\mathrm{EPR}_{K_{\mathrm{PM}}}}
    \otimes
    \ket{0\cdots 0},
\]
where
\[
    \ket{\Phi_{\mathbf{N}}}
    =
    \sum_{\bx}\sqrt{\prob{X=\bx}}\ket{\bx}\ket{\mathbf{N}\bx}.
\]
As such, the entanglement entropy is
\[
    K_{\mathrm{PM}}+H(\mathbf{N}X)
\]
across the cut $A:B$ defined above. Since $K_{\mathrm{PM}}$ is branch-independent when all shift bits are mirrored in both branches, the additive high-minus-low entropy gap is exactly the same as for the ideal state $\sum_{\bx}\sqrt{\prob{X=\bx}}\ket{\bx}\ket{\mathbf{N}\bx}$. The main text uses the randomized-encoding construction because the entropy calculation is purely classical and avoids tracking offset-induced phases.

\section{A constant-local collision-resistant hash from random-code \texorpdfstring{$\mathsf{bSVP}$}{bSVP}}
\label{app:crh}

This section gives a simple cryptographic corollary of the randomized-encoding theorem from \Cref{sec:retheorem}. The resulting hash family has bounded fan-in and bounded fan-out, and is based on the bounded syndrome-decoding / bounded shortest-vector problem for \emph{random} linear codes. Compared with earlier simple-hash constructions, the gain is that the coding assumption no longer has to come from an LDPC family; the tradeoff is weaker shrinkage because the randomized encoding introduces quadratic overhead.

\begin{definition}[Bounded syndrome decoding / bounded shortest vector]
Fix a constant $\beta\in(0,1)$. The problem $\mathsf{bSVP}_{\beta}$ for a matrix ensemble asks, given a matrix $\mathbf{N}\in\F_2^{\mu\times n}$ sampled from that ensemble, to find a nonzero vector $\bv\in\F_2^n$ such that
\[
    \mathbf{N}\bv=0
    \qquad\text{and}\qquad
    |\bv|\le \beta n.
\]
When $\mathbf{N}$ is sampled uniformly from $\F_2^{\mu\times n}$, we refer to this as the \emph{random-code} $\mathsf{bSVP}_{\beta}$ problem.
\end{definition}

For our purposes, the assumption is simply that for suitable parameters
$n,\mu=\Theta(\lambda)$ and a suitable constant $\beta>0$ satisfying
$\beta n<\mu+1$, no probabilistic (or quantum) polynomial-time adversary
can solve random-code $\mathsf{bSVP}_{\beta}$ with non-negligible
probability.  The inequality excludes a parameter range in which the
problem has a direct solution.  Indeed, when $n>\mu$, Gaussian elimination
on any $\mu+1$ columns of $\mathbf{N}$ finds a nonzero kernel vector of
weight at most $\mu+1$.  Hence, if $\beta n\geq \mu+1$, it solves
$\mathsf{bSVP}_{\beta}$ in polynomial time.

We also use a standard injective local map whose image consists of low-weight strings.

\begin{definition}[Local low-weight embedding]
An efficiently computable map
\[
    \mathsf{Exp}:\bits^k\to\bits^n
\]
is a \emph{local low-weight embedding} with relative weight $\beta$ if the following hold:
\begin{enumerate}
    \item $\mathsf{Exp}$ is injective;
    \item every output bit depends on at most $O(1)$ input bits, and every input bit influences at most $O(1)$ output bits;
    \item for every $\bx\in\bits^k$, one has $|\mathsf{Exp}(\bx)|\le \beta n/2$.
\end{enumerate}
\end{definition}

Such maps are standard in the low-complexity cryptography literature; see, for example,~\cite{EPRINT:Applebaum17,JC:AppIshKus18}. We use them as a black box.

Let $\mathbf{N}\in\F_2^{\mu\times n}$ be public and let
\[
    \mathsf{RE}_{\mathbf{N}}(\bv;\br,\bs)=(\bw,\bz)
\]
be the randomized encoding from \Cref{thm:reintro}, where now the input length is $n$ and the output length is $\mu$. Thus $\br\in\bits^{n\mu}$ and $\bs\in\bits^{\mu(n-1)}$, and the output consists of
\[
    \bw\in\bits^{n\mu}
    \qquad\text{and}\qquad
    \bz\in\bits^{n\mu}.
\]
Define the deterministic hash family
\[
    H_{\mathbf{N}}:\bits^{k+n\mu+\mu(n-1)}\to\bits^{2n\mu}
\]
by
\[
    H_{\mathbf{N}}(\bx,\br,\bs)
    :=
    \mathsf{RE}_{\mathbf{N}}(\mathsf{Exp}(\bx);\br,\bs).
\]

\begin{theorem}
Fix parameters $n,\mu,k=\Theta(\lambda)$ and constants
$\beta>0$ and $c\in(0,1)$ such that
\[
    \beta n<\mu+1
    \qquad\text{and}\qquad
    \mu\leq (1-c)k.
\]
Assume that random-code $\mathsf{bSVP}_{\beta}$ is hard for probabilistic
(or quantum) polynomial-time algorithms, and let $\mathsf{Exp}$ be any
local low-weight embedding with relative weight $\beta$. Then the family
$\{H_{\mathbf{N}}\}_{\mathbf{N}}$ is a collision-resistant hash family
with bounded fan-in and bounded fan-out.
\end{theorem}

\begin{proof}
The locality statement is immediate from the definition. Every output bit of $H_{\mathbf{N}}$ either is a $\bw$-bit or a $\bz$-bit of the randomized encoding. By \Cref{thm:reintro}, each such bit depends on at most three source bits from $(\mathsf{Exp}(\bx),\br,\bs)$. Since $\mathsf{Exp}$ itself has constant fan-in and fan-out, composing it with $\mathsf{RE}_{\mathbf{N}}$ preserves constant input locality and constant output locality.

It remains to prove collision resistance. Suppose
\[
    H_{\mathbf{N}}(\bx,\br,\bs)=H_{\mathbf{N}}(\bx',\br',\bs').
\]
By perfect decodability of the randomized encoding, decoding both sides yields
\[
    \mathbf{N}\mathsf{Exp}(\bx)=\mathbf{N}\mathsf{Exp}(\bx').
\]
Hence
\[
    \mathbf{N}(\mathsf{Exp}(\bx)\oplus \mathsf{Exp}(\bx'))=0.
\]
If $\bx\neq \bx'$, then injectivity of $\mathsf{Exp}$ implies that
\[
    \bv:=\mathsf{Exp}(\bx)\oplus \mathsf{Exp}(\bx')
\]
is nonzero. Moreover,
\[
    |\bv|
    \le
    |\mathsf{Exp}(\bx)|+|\mathsf{Exp}(\bx')|
    \le
    \beta n.
\]
Thus a nontrivial collision yields a nonzero vector $\bv$ of Hamming weight at most $\beta n$ in the kernel of $\mathbf{N}$, which solves the random-code $\mathsf{bSVP}_{\beta}$ problem.

If instead $\bx=\bx'$, then the equality of outputs forces $(\br,\bs)=(\br',\bs')$. Indeed, the $\bw$-part uniquely determines $\br$ once $\mathsf{Exp}(\bx)$ is fixed, because the map $\br\mapsto \bw$ is bijective for every input. After $\br$ is determined, the boundary convention $\bs_{i,0}=\bs_{i,n}=0$ makes each row of $\bs$ uniquely recoverable from the corresponding row of $\bz$ by the recurrence
\[
    \bs_{i,j}=\bz_{i,j}\oplus \bs_{i,j-1}\oplus \mathbf{N}_{i,j}\br_{j,i}.
\]
So equal outputs imply identical inputs. Therefore any collision finder for $H_{\mathbf{N}}$ yields an algorithm for random-code $\mathsf{bSVP}_{\beta}$.
\end{proof}

The main point of comparison is the line of work on low-complexity hash functions from~\cite{EPRINT:Applebaum17,JC:AppIshKus18}. Their constructions achieve constant locality together with \emph{linear shrinkage}, but for collision resistance the underlying coding assumption is correspondingly specialized to sparse/LDPC-like code families. By contrast, the present construction works with \emph{random} linear codes, because the randomized encoding absorbs the dense matrix--vector product into a local sampler. The price is weaker compression.

Indeed, the input length of $H_{\mathbf{N}}$ is
\[
    N_{\mathrm{in}}=k+n\mu+\mu(n-1)=k+2n\mu-\mu,
\]
and the output length is
\[
    N_{\mathrm{out}}=2n\mu.
\]
So the shrinkage is
\[
    N_{\mathrm{in}}-N_{\mathrm{out}}=k-\mu.
\]
Under the parameter conditions in the theorem,
$n,\mu,k=\Theta(\lambda)$ implies
$N_{\mathrm{in}}=\Theta(\lambda^2)$.  Moreover,
$\mu\leq(1-c)k$ gives
\[
    ck\leq k-\mu\leq k,
\]
so the hash is compressing and its shrinkage is
$\Theta(\lambda)=\Theta(\sqrt{N_{\mathrm{in}}})$. Thus the construction
should be viewed as a conceptual corollary of the randomized encoding
theorem, rather than as a replacement for the linearly shrinking
simple-hash constructions based on LDPC assumptions.

\end{document}